\documentclass[11pt]{article}
\usepackage[usenames,dvips]{color} 
\usepackage{amsthm,amssymb,amsmath,bbm,enumerate,graphicx,epsf,authblk,xcolor,setspace}
\usepackage[boxed]{algorithm2e}
\usepackage{algorithmic}
\usepackage{hyperref}
\usepackage{epsfig}
\usepackage[margin=1 in]{geometry}

\newcommand{\Debug}{0}

\ifnum \Debug = 1 
\else  \fi

\ifnum \Debug = 1 \usepackage[notref,notcite]{showkeys}
\fi

\newcommand{\keywords}[1]{\vskip2ex\par\noindent\normalfont{\bfseries Keywords: } #1}

\newcommand{\geo}[1]{\operatorname{Geo}\!\left(#1\right)}

\newcommand{\expy}[1]{\operatorname{Exp}\!\left(#1\right)}
\newcommand{\bin}[2]{\operatorname{Bin}\!\left( #1,#2 \right)}

\newcommand{\Erl}[2]{\operatorname{Gamma}\!\left( #1,#2 \right)}

\newtheorem{thm}{Theorem}  
\newtheorem{theorem}{Theorem} 

\newtheorem{lem}[thm]{Lemma}

\newtheorem{clm}[thm]{Claim}
\newtheorem{cor}[thm]{Corollary}

\newtheorem{rem}[thm]{Remark}

\newtheorem{pro}[thm]{Proposition}
\newtheorem{con}[thm]{Conjecture}
\newtheorem{pbl}{Open Problem}

\newcommand{\remove}[1]{}

\newcommand{\n}{1}

\renewcommand{\deg}{\operatorname{deg}}

\newcommand{\mix}{\operatorname{mix}}

\newcommand{\Pro}[1]{\mathbf{Pr} \left[\,#1\,\right]}

\newcommand{\Ex}[1]{\mathbf{E} \left[\,#1\,\right]}

\renewcommand{\tilde}{\widetilde}
\renewcommand{\epsilon}{\varepsilon}

\numberwithin{thm}{section}

\usepackage{multirow}
\usepackage{xspace}
\usepackage{tabularx}
\newcommand{\pts}{\operatorname{\mathbf{PtS}}}
\newcommand{\stp}{\operatorname{\mathbf{StP}}}

\newcommand{\CP}{\operatorname{\mathbf{CP}}}
\newcommand{\ptu}{\operatorname{\mathbf{PtU}}}

\usepackage[vcentermath]{youngtab}
\usepackage{multicol}
\newcommand{\BO}[1]{\mathcal{O}\!\left(#1\right)} 
\newcommand{\lo}[1]{o\!\left(#1\right)} 
\newcommand{\BOhm}[1]{\Omega\!\left(#1\right)} 
\newcommand{\lohm}[1]{\omega\!\left(#1\right)} 
\newcommand{\BT}[1]{\Theta\!\left(#1\right)} 

\title{The Dispersion Time of Random Walks on Finite Graphs\footnote{An extended abstract version of this paper appeared in 
		\newblock {\em The 31st ACM Symposium on Parallelism in Algorithms and
			Architectures}, SPAA '19, pages 103--113, New York, NY, USA, 2019. ACM. \cite{Riv}}} 

\author[1]{Nicol\'{a}s Rivera}
\author[2]{Alexandre Stauffer}
\author[1]{Thomas Sauerwald}
\author[1]{John Sylvester}
\affil[1]{\small University of Cambridge, Cambridge, United Kingdom\\
	\texttt{firstname.lastname@cl.cam.ac.uk}}
\affil[2]{\small University of Bath, Bath, United Kingdom\\
	\texttt{a.stauffer@bath.ac.uk}}

	\date{\vspace{-5ex}}
\begin{document}
	\maketitle
	\begin{abstract} 
		We study two random processes on an $n$-vertex graph inspired by the internal diffusion limited aggregation (IDLA) model. In both processes $n$ particles start from an arbitrary but fixed origin. Each particle performs a simple random walk until first encountering an unoccupied vertex, and at which point the vertex becomes occupied and the random walk terminates. In one of the processes, called \textit{Sequential-IDLA}, only one particle moves until settling and only then does the next particle start whereas in the second process, called \textit{Parallel-IDLA}, all unsettled particles move simultaneously. Our main goal is to analyze the so-called dispersion time of these processes, which is the maximum number of steps performed by any of the $n$ particles. 
		
		In order to compare the two processes, we develop a coupling which shows the dispersion time of the Parallel-IDLA stochastically dominates that of the Sequential-IDLA; however, the total number of steps performed by all particles has the same distribution in both processes. This coupling also gives us that dispersion time of Parallel-IDLA is bounded in expectation by dispersion time of the Sequential-IDLA up to a multiplicative $\log n$ factor.  Moreover, we derive asymptotic upper and lower bound on the dispersion time for several graph classes, such as cliques, cycles, binary trees, $d$-dimensional grids, hypercubes and expanders. Most of our bounds are tight up to a multiplicative constant.
	\end{abstract}
	
	\section{Introduction}

	The internal diffusion limited aggregation (IDLA) model, first introduced independently by Diaconis \& Fulton \cite{DiaFul} and Meakin \& Deutch \cite{MeDe}, is a protocol for recursively building a randomly growing subset (aggregate) of vertices of a graph. Initially, the aggregate consists of only one vertex, denoted as the \emph{origin}, and we let a particle be settled at that vertex. Then, at each step, we start a new particle from the origin and let it perform a random walk until it visits a vertex not contained in the aggregate. At this point, we say that the new particle settles at that vertex, and the vertex is added to the aggregate. We then add a new particle at the origin, and iterate this procedure over and over again. 
	
	IDLA was introduced on the infinite lattice $\mathbb{Z}^d$. Here we consider a finite connected $n$-vertex graph $G$. Note that after $n$ particles have settled, the aggregate occupies the whole of $G$. During this time, each particle performed some number of random walk steps before it settles. Clearly, this number depends on the geometry of the aggregate when the particle started moving. We define the \emph{dispersion time} as the largest number of random walk steps performed by any one of the $n$ particles before reaching an unoccupied vertex. 
	
	We refer to the above protocol as \emph{Sequential-IDLA}, in allusion to the fact that a particle cannot begin to move until the one before it settles. However, alternative scheduling protocols could be defined, in the sense that we could choose to add and move a new particle from the origin before the previous one has settled. In this way, there could be several unsettled particles moving at the same time, but they must abide by the rule that whenever an unoccupied vertex is visited, one particle must settle there. We call any process of this sort a \emph{dispersion process}. We are interested in understanding the affect of different scheduling protocols on the dispersion time. In particular consider the following protocol: start all $n$ particles from the origin at time $0$ (thus one of them will instantaneously settle at the origin). Then, all particles perform one random walk step simultaneously; if one or more particles jump to an unoccupied vertex, then one such particle settles there. Iterate this procedure until all particles have settled. We called this protocol the \emph{Parallel-IDLA}.
	
	Both dispersion processes can be regarded as a set of simple local protocols for resource allocation. Specifically, the sequential dispersion process is quite similar to a local-search based reallocation scheme from \cite{LocSearch}, where a job continues to reallocate itself to a neighbour with less load until it has found a local minimum. Furthermore, the parallel dispersion process is related to the ``QoS Load Balancing'' model~\cite{AFHS11}, a particular instance of selfish load balancing (see also \cite{B07,BHS14} for similar protocols). In the QoS model, tasks perform random walks in parallel and terminate only if they have found a resource on which the estimated processing time is acceptable according to some agent-specific threshold. Our dispersion processes can be also viewed as a spatial coordination game, where the goal is to achieve a state in which players are all making {\em distinct} choices. As mentioned in \cite{alpern2002spatial}, such games serve as a model for the dynamics in location games or habitat selection of species.
	
	Recall that the dispersion time is the maximum number of steps taken by any of the $n$ particles in either IDLA process. For the complete graph $K_n$ the Sequential-IDLA process has essentially the same dynamics as the famous coupon collector process and the dispersion time corresponds to the longest wait between collecting successive (new) coupons. Thus the discrepancy between the Sequential and Parallel dispersion times for $K_n$ measures the effect of parallelising the coupon collector process on the longest time between coupons. This motivates the study of dispersion time on different networks which we can view as a generalization of the coupon collector process. In the general setting we address the question: what is the cost of parallelising the IDLA process?  Addressing this question requires us to determine or at least estimate the Parallel {\em and} Sequential dispersion times.
	
	The total time taken by all walks, as opposed to the longest walk, is also natural to study for these models. Returning briefly to the complete graph we see that the sum of the walk lengths in the Sequential-IDLA corresponds to the time to collect all coupons - this is what is typically studied for the coupon collector. Our couplings show that for any fixed graph the sum of all walk lengths, later denoted by $W$, is the same for Parallel and Sequential IDLA. From one perspective this motivates the study of $W$ for general graphs, this is work in progress by the authors. However, in this paper we are more interested in the discrepancies between the Sequential and Parallel processes, some of which are captured by the dispersion time.  
	
	Since in IDLA particles perform random walks, both dispersion processes can be regarded as a protocol for exploring and covering an unknown network. However, as opposed to previously studied models of covering a graph with multiple random walks \cite{AlonMany,avin,CoopFrezMulti}, the length of the particles' trajectories may vary wildly in the dispersion process. This introduces strong correlations between different particles, a challenge which is not present in the cover time of multiple random walks.
	
	\subsection{Our Contributions}
	Let $\tau^v_{seq}(G)$ and $\tau^v_{par}(G)$ to denote the dispersion time of Sequential-IDLA and Parallel-IDLA on $G$ with origin $v$, respectively. The key question is how are $\tau^v_{seq}(G)$ and $\tau^v_{par}(G)$ related and is there an ordering between them. We answer this question by developing a coupling, based on ``cutting \& pasting'' particle trajectories, which we use to show the following result below. 
	\begin{theorem}[see Theorems~\ref{PStocS} and~\ref{couplog}]
		For any connected $n$-vertex graph $G$ and $v \in V(G)$, 
		$$ \tau_{seq}^v(G) \preceq \tau_{par}^v(G).$$
		Further, $$\Ex{\tau^v_{par}(G)}= \BO{\Ex{\tau^v_{seq}(G)}\cdot \log n}.$$
	\end{theorem}
	If instead we count the total number of jumps performed by all particles, then this quantity has the same distribution in both processes. Our work leaves whether $\Ex{\tau_{par}^v(G)}=\BO{\Ex{\tau_{seq}^v(G)}}$ as an open question. Note however, that Theorem \ref{complete} demonstrates that already for the clique, the Parallel-IDLA is about 30 percent slower than the Sequential-IDLA. Thus, we cannot have equality between the two processes, even though the path is an example where both processes have the same dispersion time up to lower order terms, see Theorem \ref{PathConstSame}. 
	
	In Section \ref{coupling} we introduce the continuous-time Uniform-IDLA (CTU-IDLA), a variant of the Parallel-IDLA where each particle moves at times given by its own exponential rate $1$ clock until it settles. Denote its dispersion time by $\tau_{c-unif}^v(G)$ and let $\tau_{c-seq}^v(G) $ be the dispersion time of the Sequential-IDLA run with continuous random walks. We also consider running the parallel and sequential processes with lazy walks and let $ \tau_{L-par}^v(G), \tau_{L-seq}^v(G) $ denote their dispersion times. 
	\begin{theorem}[see Theorems~\ref{lazy2normal}, \ref{conttimepar} and \ref{c-seq}] For any connected $n$-vertex graph $G$ and $v \in V(G)$, 
		\begin{align*}
		\tau_{c-unif}^v(G)&=\BT{\tau_{par}^v(G)}, \qquad \qquad \phantom{\text{and}}\qquad \qquad \tau_{L-par}^v(G) =\BT{\tau_{par}^v(G)},	\\
		\tau_{c-seq}^v(G) &=\BT{\tau_{seq}^v(G)}\phantom{,,} \qquad \qquad \text{and}\qquad \qquad\tau_{L-seq}^v(G) =\BT{\tau_{seq}^v(G)},\end{align*}hold w.h.p. and in expectation.
	\end{theorem}
	We also consider general scheduling sequences satisfying a natural condition we call ``index-repeating'' which states that if the walks are not allowed to settle and the process continues forever, then no walk will ever stop moving. We can show that greatest number of steps taken by a walk in the IDLA process according to any index-repeating schedule is stochastically dominated by the same quantity in the Parallel process (Theorem \ref{thm:UtP}).  The intuition behind Parallel-IDLA being ``slower'' than Sequential-IDLA is that, due to competition between particles trying to settle concurrently, the lengths of particle trajectories in Parallel-IDLA vary more than in Sequential-IDLA.

	Let $t_{seq}(G)=\max_{v \in V}\Ex{\tau_{seq}^v(G)}$ and $t_{par}(G)= \max_{v \in V}\Ex{\tau_{par}^v(G)}$ be the worst-case expected dispersion times over all possible origins/starting vertices in $V$. Let $t_{hit}(G)$ be the maximum among all vertices $v,w$ of the expected hitting time of a random walk from $v$ to $w$. We derive a basic but useful upper bound on the dispersion time in terms of the hitting time.
	\begin{theorem}[See Theorem~\ref{lem:general}, Corollary~\ref{cor:general}, Theorem~\ref{cycle}, Proposition~\ref{lolli}]\label{theorem:general}
		Let $G$ be any connected graph with $n$ vertices. Then, for any vertices $v\n V$, \[ \Pro{\tau_{par}^v(G) > 8 \cdot t_{hit}(G) \cdot  \log_{2} n}\leq \frac{1}{n^2}\qquad \text{and} \qquad t_{par}(G) = \BO{ t_{hit}(G)\cdot  \log n }.\]
		The same results also hold for $\tau_{seq}$ and $t_{seq}$. These results imply the following worst-case bounds:
		\begin{itemize}\itemsep0pt
			\item For any $n$-vertex graph, $t_{seq}, t_{par}=\BO{n^3\log n}$. 
			\item For any regular $n$-vertex graph, $t_{seq},t_{par}=\BO{n^2\log n }$.
		\end{itemize} 
		Moreover, the Lollipop and the cycle, respectively, are graphs matching the two bounds up to constant factors.
	\end{theorem}
	In view of the upper bound in Theorem~\ref{theorem:general} and based on the intuition that the last walk in the Sequential-IDLA should have a hard target to hit, one would expect that the worst-case hitting time provides at least an approximate lower bound on the dispersion time. This intuition turns out to be false in general, as evidenced by a certain class of bounded-degree trees (see Proposition~\ref{lowbddcounter}) which exhibits a gap of almost $\sqrt{n}$ between the hitting and dispersion time. We obtain some lower bounds based on the maximum degree $\Delta(G)$ and the $1/4$ total variation mixing time $t_{\mix}$. 
	\begin{theorem}[See Theorems~\ref{lowerbound},~\ref{treebdd} \& \ref{mixlbb}]
		Let $G$ be a connected $n$-vertex graph, then  $t_{seq}(G)=\Omega(|E|/\Delta)$. If $\Delta =\BO{|E|/n }$ then $t_{seq}(G)=\Omega(t_{\mix}) $. For any tree $T$, we have  $t_{seq}(T)= \Omega(n)$. 
	\end{theorem}
	The first two bounds are tight and the third is known to be tight up to a $\log n$ factor. The upper bound in Theorem~\ref{lem:general} matches Matthews bound for the cover time up to constant \cite[Thm. 11.2]{levin2009markov}. While Theorem~\ref{lem:general} is tight for the cycle, it turns out not to be tight for most ``well-connected'' graphs like expanders, high-dimensional grids and hypercubes. Thus as a general rule of thumb for well connected graphs the dispersion time is usually of order $t_{hit}$ and poorly connected graphs it is usually of order $t_{hit}\cdot \log n $. The behaviour of this extra log factor potentially appearing in the dispersion time contrasts with that of the log factor which may appear in the cover time.
	
	Let $ t_{hit}(\pi,S)$ denote the expected hitting time of $S\subseteq V$ by a random walk from stationarity. We provide a general framework for establishing bounds better than $\mathcal{O}(t_{hit} \log n)$ by considering certain sums of hitting times of subsets of decreasing sizes.

	\begin{theorem}[see Theorems~\ref{partialparupper} and ~\ref{sequentialprecise}, and  Corollary \ref{parlam2}]\label{precise}
		For any connected $n$-vertex graph $G$,
		\[
		t_{par}(G) \leq 60 \cdot \sum_{j=1}^{\lceil \log_{2} n \rceil} \left( t_{mix} + \max_{S \subseteq V \colon |S| \geq 2^{j-2}} t_{hit}(\pi,S) \right),
		\]
		where walks in the IDLA process are lazy. Furthermore,
		\[
		t_{seq}(G) \leq 30 \cdot \max_{1 \leq j \leq \lceil\log_2 n \rceil} \left\{ j \cdot \left( t_{mix} + \max_{S \subseteq V \colon |S| \geq 2^{j-2}} t_{hit}(\pi,S))  \right) \right\}.
		\]Consequently for any connected $n$-vertex almost-regular graph, \[t_{par}(G)=\BO{\frac{n}{1-\lambda_2}}  .\]
	\end{theorem}

	Neglecting constant factors, both upper bounds look comparable, however it is not difficult to verify that the upper bound on $t_{seq}$ is at most the upper bound on $t_{par}$, up to constants. Conversely, the gap between the two upper bounds can be shown to be at most $\mathcal{O}(\log n)$. Note that both statements recover the basic $\mathcal{O}(t_{hit} \cdot \log n)$ upper bound, but as soon as there is a sufficient speed-up for hitting times of larger sets (and the mixing time is not too large), these bounds may give a bound of $\mathcal{O}(t_{hit})$. We will see that this is indeed the case for several fundamental classes of graphs in Section \ref{fundamental}, where we apply the previous bounds, and in particular Theorem~\ref{precise}.

	\begin{table}[ht]
		\small
		\label{results table}
		\centering
		\rule{0pt}{4ex}    
		\begin{tabular}{|l|c|c|c|c|c|}
			\hline
			Graph family name  & $\text{Cover time} $ & $\text{Hitting time} $ & $\text{Mixing time}$ &\multicolumn{2}{c|}{$\text{Dispersion time}$} \\\cline{5-6} 
			& $t_{cov}$ & $t_{hit}$ & $t_{mix}$ & $t_{seq}$ & $t_{par}$  \\ \hline \hline \rule{0pt}{2.5ex}    
			path& $n^2$ & $n^2$ & $\mathcal{O}(n^2)$ & \multicolumn{2}{c|}{$\kappa_{p}\cdot  n^2\log n$}  \\ 
			\hline \rule{0pt}{2.5ex}
			cycle & $n^2/2$ & $n^2/2$ & $\mathcal{O}(n^2)$ & \multicolumn{2}{c|}{$\Theta(n^2\log n)$} \\	
			\hline \rule{0pt}{2.5ex}   
			2-dimensional grid & $\Theta(n\log^2 n)$& $\Theta(n\log n)$ &$\Theta (n)$ & $\Omega(n \log n)$ & $\BO{n\log(n)^2}$ \\ \hline \rule{0pt}{2.5ex}
			d-dimensional grid, $d>2$ & $\Theta(n \log n)$ & $ \Theta(n)$ & $\Theta(n^{2/d})$ & \multicolumn{2}{c|}{$\Theta(n)$} \\ \hline\rule{0pt}{2.5ex}
			hypercube & $\Theta(n \log n )$ &$\Theta (n)$ & $\log n \log\log n$ &\multicolumn{2}{c|}{$\Theta(n)$}\\ \hline\rule{0pt}{2.5ex}
			binary tree & $\Theta(n \log n )$ &$\Theta (n\log n)$ & $n$ &\multicolumn{2}{c|}{$\Theta(n\log(n)^2)$} \\ \hline\rule{0pt}{2.5ex}
			complete graph &$\Theta(n \log n) $& $\Theta (n)$ & $1$ & $\kappa_{cc}\cdot n  $&$(\pi^2/6)\cdot n$  \\\hline\rule{0pt}{2.5ex}
			expanders & $\Theta(n \log n)$ & $\Theta (n)$ & $\BO{\log n}$ & \multicolumn{2}{c|}{$\Theta(n)$ }  \\\hline
		\end{tabular}\caption{The last two columns summarize our results, the first three columns are for comparison. The constant $\kappa_{cc}$ above has an explicit formula given by Lemma \ref{bren} and it evaluates to roughly $1.255$, to be contrasted with $\pi^2/6 \approx 1.644$. The constant $\kappa_p$ is a non-explicit, though specified in Section \ref{fundamental}. Simulations run by Nikolaus Howe (student) suggest $\kappa_p\approx 0.6\dots$.}
		\label{tbl:results}
	\end{table}

	In Section \ref{fundamental} we calculate the dispersion times in several fundamental networks. Table \ref{tbl:results} summarises our results and shows that we can determine the expected dispersion times up to multiplicative constant factors in all graphs apart from the $2$-dimensional grid, where there is a discrepancy of order $\log n$ between the lower and upper bounds. This remains an interesting open problem which seems to require very detailed knowledge of the shape of the aggregate on a {\em finite} box/tori. As discussed in Section \ref{relwork} below, this is a non-trivial problem even in the {\em infinite} 2d-grid. For many other graphs we obtain the correct upper bound from one of our results in Section \ref{bounds} and then have to find a matching lower bound by hand. One particularly involved case not captured by our general results is the binary tree, where a tailored analysis reveals a (relatively large) dispersion time of $\Theta(n \log^2 n)=\Theta(t_{hit}\cdot \log n)$, see Theorem~\ref{thm:binarytree}.
	
	\subsection{Techniques Used}
	The first tool we invent to analyse these processes is the Cut \& Paste bijection which maps between the histories of IDLA processes. The bijection allows us to couple the dispersion times of the Parallel-IDLA to those of the Sequential-IDLA and other variants such as Uniform-IDLA (where at each step a random unsettled particle moves), as well as IDLA processes with lazy or continuous-time walks. Bounding dispersion times via these other variants is useful for avoiding issues such as periodicity or simultaneous arrivals at unoccupied vertices. At a base level the stochastic domination of $\tau_{seq}^v$ by $\tau_{par}^v$ means we can sandwich both quantities with a bound on $\tau_{par}^v$ from above and on $\tau_{seq}^v$ from below. Another useful way describe dispersion time is in terms of hitting times of sets by multiple random walks. In particular we present two different upper bounds on $\tau_{par}^v$ and $\tau_{seq}^v $ in terms of hitting times of sets. We also prove a lower bound on $\tau_{seq}^v $ by the mixing time, this comes from the relationship between the mixing time and the hitting time of large sets.
	
	Although the Sequential and Parallel IDLA processes are closely related, the different sources of dependence arising from the contrasting scheduling protocols provide several challenges. In the Sequential-IDLA interaction between the walkers comes via the configuration of vertices settled by the previous walks. This can make proving a tight lower bound on $\tau_{seq}^v$ tricky and often some knowledge of the geometry of the aggregate after a certain time is helpful. What is needed are results reminiscent of the ``shape theorems'' discussed in Section \ref{relwork} below. This requirement for detailed knowledge of the aggregate appears to be crucial in achieving a tight lower bound on $\tau_{seq}^v$ for the binary tree and 2-dimensional grid. In comparison with the Sequential-IDLA interactions are less passive in the Parallel-IDLA as particles jostle to be the first to settle a vertex. This interaction can increase the length of the longest walk as is witnessed by the Cut \& Paste bijection.
	
	\subsection{Related Work}\label{relwork} 
	As pointed out by Diaconis \& Fulton \cite{DiaFul}, there are several mathematical reasons for studying IDLA, including using it to take a product of sets - a special case of the ``smash product''. 
	The limit shape of the aggregate on $\mathbb{Z}^d$ was first studied by Lawler, Bramson and Griffeath \cite{LawBramGriff} who showed that, after adding $n$ particles and properly rescaling the aggregate by $n^{1/d}$, in the limit as $n\to \infty$ this converges to an Euclidean ball. There has been a series of improvements to this ``shape theorem'' of \cite{LawBramGriff}, by bounding the rate of convergence to the euclidean ball. The first refinement was made by Lawler \cite{Law95} and the state of the art was achieved recently by two independent groups of authors \cite{AssGau1,AssGau2,AssGau3,Sheff1,Sheff2}. 
	Several authors have also proved shape theorems on other infinite graphs and groups including combs, $d$-ary trees, non-amenable groups and Bernoulli percolation on $\mathbb{Z}^d$ \cite{HusSav,IDLAgroups,IDLAnonamengraphs,Shel,CopLucCyrAri}. In all of these cases the limit shape is always a ball with respect to the underlying graph metric. Limit shapes in $\mathbb{Z}^d$ for other variants of IDLA have also been established. These variations include using non-standard random walks such as for drifted \cite{lucas} and  cookie walks \cite{Cookie} or starting the walks from different positions \cite{unif}. The time for the process started with some initial aggregate to ``forget'' this starting state has also been studied \cite{LevSil,SilNew}. 
	
	One model where interaction between particles prevents settling at a site is a two-type particle system called ``Oil and Water'' where particles of opposite types displace each other \cite{OilWater}. There have been some papers on models related to the Parking function of a graph where cars drive randomly around a graph searching for vacant spots  \cite{park,GoldPrz}. More commonly, however, interaction is directly between particles and not with the host graph such as predator prey/coalescing models \cite{CoopFrezMulti}. The problem of uniformly distributing $n$  non-communicating memoryless particles across $n$ unoccupied sites is also considered from a game theoretic perspective \cite{alpern2002spatial}. 
	
	Other models related to IDLA include rotor-router aggregation, chip firing, Abelian sandpile models and activated random walks \cite{Chip,LaplacianGrow,arw}. Many of these interacting particle systems satisfy a so-called ``least action principle'' which is key to their analysis. Such a principle roughly states that the natural behavior of the system is in a sense optimal and, if the process is perturbed, then the outcome will have a higher energy. One may try to find a least action principle for Sequential-IDLA by conjecturing that if we allow that a random walk sometimes does not settle when visiting an unoccupied vertex (thereby performing more random walk steps), then this could only delay the dispersion time. However, we show in Proposition \ref{leastact} that this is not the case. In particular, we give a graph for which the dispersion time decreases if one allows some particles to perform more random walks steps.

	To the best of our knowledge, the dispersion time and IDLA on a finite graph has not been studied before. Moore and Machta consider running IDLA walks synchronously for the purposes of simulating the limit shape \cite{MoMaParallel} in parallel models of computation, however their results don't appear to overlap with ours. Simulating the process efficiently has also been studied more recently \cite{FriLev}. Thacker and Volkov \cite{ThaVol} study a border DLA based growth model on finite graphs and investigate how long until the aggregate grown from a fixed origin hits a fixed boundary.

	\section{Preliminaries} \label{prelim}
	
	Throughout $G=(V,E)$ will denote an undirected, unweighted, connected graph with $n$ vertices. Let $\Delta(G)$ denote the maximum degree of $G$. We say that a graph $G$ is almost-regular if the ratio between maximum degree and minimum degree is bounded from above by a constant.
	
	To recap we let $\tau^v_{par}(G)$ denote the dispersion time of the Parallel-IDLA process on $G$ started from $v$, that is the first iteration at which every vertex hosts (exactly) one particle. Similarly $\tau^v_{seq}(G)$ denotes the dispersion time of the Sequential-IDLA process on $G$ started from $v$, that is the longest time it takes a single particle to settle. Let $t_{seq}(G)=\max_{v \in V}\Ex{\tau_{seq}^v(G)}$ and $t_{par}(G)= \max_{v \in V}\Ex{\tau_{par}^v(G)}$. We shall drop the dependence on $G$ from our notation when the graph is clear from the context. 
	
	Further, let $t_{hit}(u,v) = \Ex{ \tau_{hit}(u,v) }$, where $\tau_{hit}(u,v)$ is the time for a random walk to reach $v$ from $u$. Let $t_{hit}(G) := \max_{u,v\in V(G)} t_{hit}(u,v)$. For a probability distribution $\mu$ on $V$ and a set $S\subset V$ let $t_{hit}(\mu,S)$ denote the expected time for the walk starting from $\mu$ to hit any vertex in $S$. 
	
	Thanks to our results relating lazy and non-lazy walks (Theorem \ref{lazy2normal}), we can conveniently switch between the two models at the cost of a constant factor (under some mild additional conditions this factor is $2+o(1)$), thus walks may be lazy. We use $P$ to denote the transition matrix of the non-lazy walk (and $\tilde{P} = \left(I + P \right)/2$ for the lazy walk). We also use $p^t_{u,v}$ to denote the probability a random walk goes from $u$ to $v$ in $t$ steps (and $\tilde{p}^t_{u,v}$ respectively for the lazy walk). We let $t_{\mix}=\min_{t\geq 1}\left\{t:\;\max_{x\in V}\sum_{y\in V}\left|\tilde{p}^{t}_{x,y}-\pi(y) \right|\leq 1/e\right\}$ denote the mixing time of $G$.  
	
	Some of the dispersion results in the paper hold in expectation, some hold w.h.p.\ (with probability $1-o(1)$) and others hold in both senses. One does not necessarily imply the other, in particular Proposition \ref{conccounter} show there are graphs where neither dispersion time concentrates. 
	\paragraph{Road Map.}
	The rest of this paper is organized as follows. We first present some general upper and lower bounds in Section~\ref{bounds} before turning to the more involved coupling proofs in Section~\ref{coupling}. In Section~\ref{fundamental} we apply the results from Section~\ref{bounds} and Section~\ref{coupling} to specific networks completing the results in Table~\ref{results table}, for some of these networks a more refined analysis is required. We conclude the paper in Section~\ref{sec:conclusion} with a summary of our results and some open problems.
	
	\section{General Bounds} \label{bounds}
	\subsection{Upper Bounds}
	The first upper bound we present holds for any graph and only requires knowledge of the maximum hitting time of a random walk between two vertices. Although this result can be also recovered from the more general Theorem~\ref{partialparupper}, it serves as a good ``warm-up''. 
	
	\begin{thm}\label{lem:general}
		Let $G$ be any connected graph with $n$ vertices. Then for any $v \in V$, \[ \Pro{\tau_{par}^v(G) > 8 \cdot t_{hit}(G) \cdot  \log_{2} n}\leq \frac{1}{n^2}\qquad \text{and} \qquad t_{par}(G) = \BO{ t_{hit}(G)\cdot  \log n }.\]
		The same results also hold for $\tau_{seq}^v$ and $t_{seq}$.
	\end{thm}
	\begin{proof}
		To begin sample $n$ random walks of length $T=8t_{hit}(G)\log_2 n$ starting from the origin, then w.p. at least $1-n^{-2}$ all of these walks have covered all the vertices of the graph. To see this note the probability a single walk of length $2t_{hit}$ visits $u \in V$ is at least $1/2$ by Markov's inequality, thus by the Markov property $u$ is visited in time $T$ w.p. $1-n^{-4}$. Thus, by a union bound, in time $T$ one walk covers the graph w.p. at least $1-n^{-3}$ and all walks cover the graph w.p. at least $1-n^{-2}$. 
		
		Now, we run the Parallel-IDLA process by using the $n$ sampled walks, thus each particle follows a predetermined trajectory. Since all the $n$ walks cover the graph, it follows that all the particles have to settle by time $T=8t_{hit}(G)\log_2 n$ with probability at least $1-n^{-2}$. To obtain the result in expectation, divide the time in phases of $8t_{hit}(G)\log_2 n$ time-steps, then the number of phases needed to finish the process is stochastically dominated a geometric random variable of mean $1/(1-n^{-1})$ concluding that $\Ex{\tau_{par}^v(G)} = \mathcal{O}(t_{hit}(G)\log n)$. Since this holds for any $v\in V$ it follows that $t_{par}=\mathcal{O}(t_{hit}(G)\log n )$. The same results holds for $\tau_{seq}^v$ and $t_{seq}$ due to Theorem~\ref{PStocS}.
	\end{proof}

	This simple bound is actually tight in many cases, see Table~\ref{tbl:results}. The next result is a simple consequence, yet it provides the correct asymptotic worst-case bounds for the dispersion time.
	
	\begin{cor}[General quantitative bounds on graphs]\label{cor:general}$\qquad $

		\begin{itemize}\itemsep0pt
			\item For any $n$-vertex graph, $t_{seq}, t_{par}=\BO{n^3\log n}$. 
			\item For any regular $n$-vertex graph, $t_{seq},t_{par}=\BO{n^2\log n}$.
		\end{itemize} 
	\end{cor}
	\begin{proof}This follows from Theorem \ref{lem:general} and the bounds on $t_{hit}$ in \cite[Thm. 2.1]{lovasz1993random}.
	\end{proof} 
	
	Notice these bounds exceed the corresponding upper bounds on the cover time \cite[Thm. 6.12, Thm. 6.15]{aldousfill} by a $\log n$-factor. Both bounds above are sharp up to a multiplicative constant as witnessed by the lollipop and the cycle respectively, see Proposition \ref{lolli} and Theorem \ref{cycle} respectively. In fact for any fixed $r<\infty$ one can construct a family of $r$-regular graphs for which the second bound above is tight. For example when $r=3$ one can iteratively augment an even cycle by adding an edge between two vertices of degree two who are at distance two to obtain a 3-regular graph with the same asymptotic dispersion time as the cycle.

	\subsubsection{General Bounds in Terms of Hitting Times of Sets}
	In this section we achieve more refined bounds by considering hitting times of sets as opposed to vertices. To avoid periodicity related issues we assume the trajectory of the particles is a lazy random walk. As shown in Theorem~\ref{lazy2normal}, the parallel or sequential dispersion times with lazy walk are equivalent to their non-lazy counterparts up to constant factors, thus any results established for the dispersion time with lazy walks also apply for non-lazy walk (up to a constant factor) and vice versa.  Define $\tau_{par}^{v}(G,k)$ to be the first time (from worst case start vertex) that less than $ 2^k -1  $ vertices are left to be settled in the Parallel-IDLA, and let $t_{par}^{k
	}(G)= \max_{v \in V} \Ex{\tau_{par}^{v}(G,k) } $ denote the worst-case expectation. Clearly $\tau_{par}^{v}(G,1) = \tau_{par}^v(G)$, which is the standard parallel dispersion time.  
	\begin{thm}\label{partialparupper}
		Consider the Parallel-IDLA process with lazy walks. Then, for any connected $n$-vertex graph and any $k\geq 1 $, we have
		\[
		t_{par}^k(G) \leq 60 \cdot \sum_{j=k}^{\lceil \log_{2} n \rceil} \left( t_{mix} + \max_{S \subseteq V \colon |S| \geq 2^{j-2}} t_{hit}(\pi,S) \right).
		\]
	\end{thm}
	One consequence of this theorem for $k=\log_2 n - 1$ is that within $\mathcal{O}(t_{mix})$ steps, at least $n/2$ random walks are settled (this follows since by the duality between hitting time of large sets and mixing time \cite{PerSou}, $\max_{S \subseteq V \colon |S| \geq n/4} t_{hit}(\pi,S)  =  \mathcal{O}(t_{mix})$.
	\begin{rem}Note that the upper bound can be estimated directly to be at most $60\lceil \log_{2} n \rceil \cdot \left( t_{mix} + t_{hit} \right) \leq 120 \lceil \log_{2} n \rceil \cdot t_{hit}$, so this bound is (up to a multiplicative constant) a refinement of Theorem \ref{lem:general}.
	\end{rem}

	\begin{proof}[Proof of Theorem \ref{partialparupper}]
		We divide the process into $\log_2 n$ phases which are labelled in reverse order $\lceil \log_{2} n \rceil, \lceil \log_{2} n \rceil-1,\ldots,2,1$. Phase $j$ starts as soon as the number of unsettled walks $k$ satisfies $k \in [2^{j-1},{2}^{j})$. It could be case that the number of unsettled walks more than halves in one step and phase $j$ is skipped, for now assume this is not the case. Let $t$ be the first time step at the beginning of phase $j$, and let $S\subseteq V$ be the set of unoccupied vertices at time $t$, thus $|S|=k$. Consider $k$ random walks moving independently and having no interaction with the unsettled vertices, then let $\tau_j$ be the (random) time such that no subset $S'$ of $S$ with size at least $k/2$ is visited by any less than $k/2$ of these walks. We now argue by contradiction that $\tau_j$ stochastically dominates the length of phase $j$. Suppose the number of unsettled walks is still at least $ k/2$ at step $t+\tau_j$. Hence there exists still a subset $S'$ of unoccupied vertices with size at least $k/2$ at step $t+\tau_j$. We know that at least $k/2$ of the walks would hit at least one vertex of this set $S'$. Thus all these walks must terminate earlier, as otherwise the vertices in $S'$ cannot all be unoccupied at step $t+\tau_j$, however, in this case we have a contradiction to the assumption that at least $k/2$ of the walks are still unsettled.

		We will now bound $\Ex{\tau_j}$ from above. Consider first a fixed random walk and a fixed set $S' \subseteq S$ of size at least $k/2$. The probability that a fixed random walk does not hit $S'$ within $30 \cdot (t_{mix} + t_{hit}(\pi,S'))$ steps is at most $(1/2)^6$, this follows easily from the fact that after $5t_{mix}$ time, with probability at least $1-e^{-1}$, we can couple the Markov chain with the stationary distribution  (e.g. Lemma A.5. in  \cite{kanade2016coalescence}), and then, given that the coupling is successful Markov's inequality gives us that with probability at most $1/5$ we do not hit $S'$, thus, the probability we do not hit $S'$ in  $5 \cdot (t_{mix} + t_{hit}(\pi,S'))$ steps is at most $e^{-1}+(1-e^{-1})(1/5) < 1/2$, and thus after $6$ time-intervals of length $5 \cdot (t_{mix} + t_{hit}(\pi,S'))$ the probability the walk does not hit $S'$ is at most $(1/2)^6$.
		
		Hence the probability that at least $k/2$ of the $k$ walks do not hit the set $S'$ is at most
		\[
		\binom{k}{k/2} \cdot \left(\frac{1}{2^6}\right)^{-k/2} \leq 2^k \cdot 2^{-3k}.
		\]
		Taking the Union bound over all possible $\binom{k}{k/2} \leq 2^k$ subsets of $S$ which are of size at least $k/2$, it follows that the probability that there exists a subset $S$ of the unoccupied vertices of size at least $k/2$ such that at least $k/2$ of the walks do not hit the set $S$ is at most
		\[
		2^k \cdot 2^k \cdot 2^{-3k} \leq 2^{-k} \leq 1/2.
		\]
		Hence the expected time the process spends in phase $j$ (assuming that we reach this phase and do not skip it) is at most
		\[
		2\cdot 30 \cdot (t_{mix} + \max_{S \subseteq V \colon |S| \geq 2^{j-2}} t_{hit}(\pi,S)).
		\]Summing up these contributions from $k$ to $\lceil \log_2 n\rceil$ yields the result.
	\end{proof}
	
	\begin{cor}\label{parlam2}
		Let $G$ be a connected $n$-vertex almost regular-graph. Then, $t_{par}(G)=\BO{n/(1-\lambda_2)}  $.
	\end{cor}
	\begin{proof}
		By Lemma \ref{setestimate} states that 
		$t_{hit}(v,S) =\BO{\frac{n(1+\lceil \log |S| \rceil)}{(1-\lambda_2) |S|}}$ holds for any $S\subset V$, $v\in V$. We also have $t_{\mix} = \BO{\frac{\log n}{1-\lambda_2} } $ by \cite[Thm.\ 12.3]{levin2009markov}. Plugging these estimates into Theorem \ref{partialparupper} yields
		\[
		t_{par}(G) \leq \BO{1} \cdot \sum_{j=1}^{\lceil \log_{2} n \rceil} \left(\frac{\log n}{1-\lambda_2} + n\cdot \frac{1+\log 2^{j-2})}{(1-\lambda_2)2^{j-2}}\right)=\BO{\frac{n}{1-\lambda_2}}.
		\]The result follows since lazy and non-lazy dispersion times are equivalent up to a constant factor by Theorem \ref{lazy2normal}. 
	\end{proof}
	Let us now turn to the sequential process, where we can derive a similar bound, which turns out to be slightly stronger.
	\begin{thm}\label{sequentialprecise} For any $n$ vertex graph $G$, we have
		\[
		t_{seq}(G) \leq 30 \cdot \max_{1 \leq j \leq \lceil\log_2 n \rceil} \left\{ j \cdot \left( t_{mix} + \max_{S \subseteq V \colon |S| \geq 2^{j-2}} t_{hit}(\pi,S))  \right) \right\}.
		\]
	\end{thm}
	\begin{proof}Since in the sequential process only one walk moves at a time we can couple simple and lazy walks so that the dispersion time with simple walks is always less than with lazy walks. Thus we can assume the walk is lazy. 
		Fix a time $\tau$ to be determined later.
		Consider the $(n-k)$-th walk in the Sequential-IDLA, when there are still $k$ unoccupied vertices. It was argued in the proof of Theorem~\ref{partialparupper} that the probability the random walk does not hit a set $S$ of size $k$ within $5 (t_{mix} + \max_{S \subseteq V \colon |S|=k} t_{hit}(\pi,S))$ time steps is at most $1/2$ regardless of the initial vertex $v$. Denote 
		$$q(k) = \left\lfloor \frac{\tau}{5\cdot (t_{mix} + \max_{S \subseteq V \colon |S|=k} t_{hit}(\pi,S))}\right\rfloor,$$
		hence the probability that the random walk does not succeed within $\tau$ steps (assuming $\tau$ is large enough) is at most $2^{-q(k)}$.	Thus by the Union bound, the probability that at least one of the $n$ walks do not succeed is at most $\sum_{k=1}^n 2^{-q(k)}.$ By dividing the sum into $\lceil \log_2 n\rceil$ buckets of sizes (at most) $1,2,\ldots, 2^m,\ldots, 2^{\lceil \log_2 n \rceil}$, and using monotonicity of hitting times, it follows that the above term is at most
		\[
		\sum_{j=1}^{\lceil \log_2 n \rceil} 2^j \cdot \exp\left(- \frac{\tau \log 2}{5 \cdot (t_{mix} + \max_{S \subseteq V \colon |S| \geq 2^{j-2}} t_{hit}(\pi,S))}		\right).
		\]
		Next observe that we need to ensure that for every $j$ it holds that 
		\[
		\tau \geq j \cdot 5 \left( t_{mix} + \max_{S \subseteq V \colon |S| \geq 2^{j-2}} t_{hit}(\pi,S))  \right), 
		\]
		otherwise just a single addend above is larger than $1$. However, if we just choose
		\[
		\tau := 3\max_{1 \leq j \leq \log_2 n} \left\{ j \cdot 5 \cdot \left( t_{mix} + \max_{S \subseteq V \colon |S| \geq 2^{j-2}} t_{hit}(\pi,S))  \right) \right\},
		\]
		then we see that the total sum in the Union bound expression is at most $1/2$, and we can conclude that with probability at least $1/2$ none of the $k$ walks takes more than $\tau$ steps. Repeating the argument that the probability that one walk take more than $m\tau$ steps is at most $2^{-m}$ gives the result.
	\end{proof}
	
	It can be checked that the bounds of Theorem~\ref{sequentialprecise} are (up to constant) potentially better than the bounds of Theorem~\ref{partialparupper} up to a $\log n$ factor.
	
	Bounds on the expected hitting time of sets can be obtained by analyzing return probabilities, in some situations these bounds are very tight. Since those bounds are more related to Markov chains properties than the  IDLA process, and in order to keep the analysis of the IDLA process as clean as possible, we do not provide those bounds here, but in Appendix~\ref{section:BoundsHitSet}. These bounds can be applied either in Theorem~\ref{partialparupper} and Theorem~\ref{sequentialprecise}, but also in directly for specific graph families.

	\subsection{Lower Bounds}
	\begin{thm}\label{lowerbound}
		Let $G$ be a connected $n$-vertex graph with maximum degree $\Delta$, then  $t_{seq}(G)=\Omega(|E|/\Delta)$. Hence in particular, $\Omega(n)$ is a lower bound for almost-regular graphs.
	\end{thm}
	\begin{proof}
		We will analyse the Sequential-IDLA process and lower bound the time it takes for the last walk to find a free site.
		
		Recall that for any pair of vertices $u,v \in V$, $t_{com}(u,v)=t_{hit}(u,v)+t_{hit}(v,u)$ is the commute time between $u$ and $v$.
		By \cite[Cor. 2.5]{lovasz1993random} there is an ordering of the $n$ vertices so that if $u$ precedes $v$, then $t_{hit}(u,v) \leq t_{hit}(v,u)$. Let us take the vertex $w$ as the origin of the dispersion process so that for any other vertex $v$, we have $t_{hit}(w,v) \geq t_{hit}(v,w)$. Hence for every vertex $v$,
		\[
		t_{hit}(w,v) \geq 1/2 \cdot t_{com}(w,v).
		\]
		Let $ R(u,v)$ be the effective resistance between $u$ and $v$ and note that $R(w,v) \geq 1/\deg(w) + 1/\deg(v) \geq 2/\Delta$. Hence $t_{com}(w,v) = 2|E| \cdot R(w,v) = \Omega(|E|/(\Delta+1))$ by the commute time identity \cite[Prop. 10.6]{levin2009markov}. It follows that, in expectation, the last walk in the Sequential-IDLA takes $\Omega(|E|/\Delta)$ steps.
	\end{proof}
	Theorem \ref{complete} shows this is tight up to constant when $G$ is the complete graph $K_n$. We also present a refined lower bound for trees. 
	
	\begin{thm}\label{treebdd}Let $T$ be any $n$-vertex tree, then $t_{seq}(T)\geq 2n-3. $
	\end{thm}	
	\begin{proof}
		If an IDLA process started from any vertex of $T$ the last vertex settled by must be a leaf. Call the last vertex $v$ which is connected to $T$ by one edge $\{u,v\}$. Thus the expected time taken by the last walk to settle is at least the expected time $t_{hit}(u,v)$ to cross the edge $\{u,v\}$. The Essential edge Lemma \cite[Lem. 5.1]{aldousfill} states that $H(u,v) = 2|A(u,v)|-1$ where $A(u,v)$ is the component of $T$ containing $u$ after the removal of $\{u,v \}$. Since $|A(u,v)| \geq n-1$, the proof is complete.
	\end{proof}
	Let $S_n$ be the $n$-vertex star and notice that $t_{seq}(S_n)= 2t_{seq}(K_n)\approx 2.6 n $ by Theorem \ref{complete}. This shows Theorem \ref{treebdd} is tight up to a small multiplicative constant.
	
	\begin{rem}\label{rem:lowbddcounter} It would be natural to hope the lower bound $t_{seq}= \Omega\left( t_{hit}\right)$ should hold since one would expect the vertices with largest hitting times to be explored later by the sequential process and thus contribute to the dispersion time. Proposition refutes this by exhibiting a graph where $t_{seq}$ is a $poly(n)$-factor smaller than $t_{hit} $.
	\end{rem} 
	For a graph $G$ let $\Phi$ be the conductance of $G$ and let $\lambda_2$ and $t_{\mix}$ be the second largest eigenvalue and mixing time associated with the lazy random walk on $G$ respectively. The following lower bound is tight up to a $\log n$ factor as witnessed by the cycle, Theorem \ref{cycle}.
	
	\begin{pro}\label{mixlbb}
		Let $G$ be a graph satisfying $\Delta =\BO{|E|/n }$ . Then there exits a $v\in V$ such that $\BOhm{n}$ walks in the Sequential-IDLA process from $v$ talk time $\BOhm{t_{\mix}}$ to settle w.h.p., consequently
		\[
		t_{seq}(G)= \Omega( t_{\mix})   = \Omega \left( \frac{1}{1-\lambda_2} \right)= \Omega \left( \frac{1}{\Phi} \right).
		\]

	\end{pro}
	\begin{proof}
		By the characterization of mixing times by hitting times of large sets \cite{PerSou}, for all reversible lazy random walks
		\begin{equation}\label{hitmixcharactor}
		t_{\mix} \leq c \max_{u, A: \pi(A) > 1/3} t_{hit}(u,A),
		\end{equation}
		where $c<\infty $ is a universal constant, which can be assumed to be greater than 1. Let $u$ and $A$ be a vertex and a set that together maximize the above expectation $t_{hit}(u,A)$. Let $r:=\Delta\cdot n/|E| <\infty $ and observe that $|A|\Delta > 2|E|/3 $ and thus $|A|> n/(3 r) $. Consider now a simple random walk of length $\ell:= t_{\mix}/(120\cdot r\cdot c)$. For every vertex $v \in V$, let $p_v$ be the probability that a random walk starting from $v$ hits the set $A$ within $\ell$ steps. Note that there must be at least one vertex $w \in V$ such that $p_w <1/(12\cdot r)$ since otherwise the expected time to hit $A$ is less than $t_{\mix}/(c\cdot 10)$ for all vertices $v$, contradicting \eqref{hitmixcharactor}. 
		
		Let $X$ be the number of walks from $w$ that take time less than $\ell$ to hit $A$ (if there were not allowed to settle before this). As these walks are independent it follows that $X $ is distributed as $\bin{n}{p_w}$, thus $\Ex{X} <n/(12r) $ and $ \Pro{X \geq n/(6r)} \leq e^{-\Omega(np_{w})} $ by the Chernoff bound. Thus w.h.p.\ at least $(1-1/(6r))n$ of the $n$ walks will take time at least $\ell$ to reach $A$, thus at least $n/(6r) $ walks which settling in $A$ take time at least $\ell$. Hence 
		\[
		t_{seq} = \Omega(t_{\mix}),
		\]
		proving the result. Then using the fact that $t_{mix}\geq \left(1/(1-\lambda_2) -1 \right)\log(1/e)$ \cite[Thm. 12.4]{levin2009markov} and then the fact that we need at least one step to mix gives $t_{mix} = \Omega( \frac{1}{1-\lambda_2})$. Cheeger's inequality \cite[Thm. 13.14]{levin2009markov}, which states $\frac{1}{1-\lambda_2} = \Omega( \frac{1}{\Phi})$, completes the proof.
	\end{proof}

	The following bound will be of use in Section \ref{coupling} as although rather weak it holds w.h.p. for any start vertex. The proof appears in a different context \cite{TightBounds2011} but we reproduce it here for completeness.
	
	\begin{lem}\label{lemma:whpLB}
		For any connected $n$-vertex graph and $v \in V(G)$, if $n$ is large enough, it holds that $\tau^v_{seq} > \frac{1}{14}\log n$ with probability at least $1-e^{-n^{1/2}}$.
	\end{lem}
	\begin{proof}
		Consider the following process: run $n$ independent random walks on $G$ starting from $v$, and we stop them at time $L = c \log n$, with $c = 1/14$. Denote by $\mathcal C$ the set of vertices that are hit by at least one of those random walks. A simple coupling argument shows that $\Pro{\mathcal C \neq V} \leq \Pro{\tau^v_{seq} > L}$, and thus we will prove that $\Pro{\mathcal C \neq V} \geq 1-e^{-n^{1/2}}.$
		
		For any of the $n$ walks, denote by $\mathcal C_i$ the set of vertices covered by the $i$-th walker in the first $L$ steps. Hence $\mathcal C = \cup_{i=1}^n \mathcal C_i$.  Denote $U = \{u \in V: \Pro{u \in \mathcal C_i} \geq \frac{3}{2}\frac{L}{n}\}$, and note $U$ is independent of $i$. Hence for any $i$,\begin{align}
		L \geq \Ex{|\mathcal C_i|} = \sum_{u \in V}\Pro{u \in \mathcal C_i} \geq \sum_{u \in U} \frac{3}{2}\frac{L}{n} \geq |U|\frac{3}{2}\frac{L}{n},\nonumber
		\end{align}
		therefore $|U|\leq \frac{2}{3}n$. Denote by $\mathcal D = V \setminus \mathcal C$ the set of uncovered vertices, then
		
		\begin{align}
		\Ex{|\mathcal D|} \geq \sum_{u \in V\setminus U} \Pro{u \in \mathcal D} \geq \sum_{u \in V \setminus U} (1-\Pro{u \in C_1})^n \geq \frac{3}{n}\left(1-\frac{3L}{2n}\right)^n \geq  \frac{n}{3}\cdot e^{ - 3L/(2-3L/n)},\nonumber 
		\end{align}
		where in the last step we use the bound $e^{-x/(1-x)} \leq 1-x$ for $|x|<1$. Hence, we deduce that $\Ex{|\mathcal D|} \geq \frac{n}{3} n^{-\frac{3c}{2}(1+o(1))} \geq \frac{1}{3}n^{1-3c}.$
		
		Finally, we will prove that with probability at least $1-e^{-n^{1/2}}$ it holds that $|\mathcal D|$ is greater than $\frac{1}2 \Ex{|\mathcal D|}$, and then $|\mathcal D|> 0$. Note that $|\mathcal D|$ depends on the trajectory of the $n$ random walks and changing one of them changes $|\mathcal D|$ in at most $L+1$ values. Therefore by the method of bounded differences we have
		
		\[
		\Pro{|\mathcal D|-\Ex{|\mathcal D|}<-\frac{1}{2}\Ex{|\mathcal D|}} \leq \exp\left(-\frac{\Ex{|\mathcal D|}^2}{2n(L+1)^2} \right)\nonumber\leq \exp\left(-\frac{\left(\frac{1}{3}n^{1-3c} \right)^2}{2n(L+1)^2} \right),\]recalling $L = c \log n$ gives that $\Pro{|\mathcal D|-\Ex{|\mathcal D|}<-\frac{1}{2}\Ex{|\mathcal D|}}$ is at most
		\[\exp\left( - \frac{n^{2-6c}}{18n (c\log n)^2(1+o(1))} \right)\leq\exp\left( - \frac{n^{1-13c/2}}{20 c^2} \right)\nonumber \leq \exp\left(-n^{1-7c}\right) =e^{-n^{1/2}}, \nonumber\]
		since $c = 1/14$. The result follows as $\Ex{|\mathcal D|} \geq \frac{1}3 n^{1-3/14}>2$.
	\end{proof}

	\section{Coupling and Stochastic Domination}\label{coupling}

	In this section we shall prove the following stochastic domination using a coupling.
	\begin{thm}\label{PStocS}Let $G$ be a finite graph and $v \in V(G)$. Then 
		$$ \tau_{seq}^v(G) \preceq \tau_{par}^v(G).$$
		
	\end{thm}An immediate corollary of this is the relation $\Ex{\tau_{seq}^v(G)} \leq \Ex{\tau_{par}^v(G)}$, we also prove the reverse inequality up to $\log n$ factors. 
	\begin{thm}\label{couplog}Let $G$ be a finite graph and $v \in V(G)$. Then 
		$$\Ex{\tau_{par}^v(G)} = \BO{ \Ex{\tau_{seq}^v(G)}\cdot\log  n }.$$
	\end{thm}
	
	We define the lazy Sequential/Parallel-IDLA to be the Sequential/Parallel-IDLA with the particles moving according to a lazy (instead of simple) random walk. Let $\tau_{L-seq}^v(G)$ be the dispersion time of the lazy Sequential-IDLA on $G$ starting from $v$, and $\tau_{L-par}^v(G)$ be the analogous quantity for the lazy Parallel-IDLA. The relation between the lazy and standard IDLA dispersion times is given in the following theorem.
	
	\begin{thm}\label{lazy2normal}Let $G$ and $v\in V(G)$. Then the following holds w.h.p.\ and in expectation\[\tau_{L-seq}^v(G) =\BT{\tau_{seq}^v(G)} \qquad \qquad \text{and}\qquad \qquad \tau_{L-par}^v(G) =\BT{\tau_{par}^v(G)}.\]
		Additionally, if there exits some $\ell = \lohm{\log n}$ such that $\Pro{\tau_{par}^v(G) \leq \ell} \leq 1/\ell $ then	
		\[
		\tau_{L-seq}^v(G) = (2+o(1))\cdot \tau_{seq}^v(G) \qquad \qquad \text{and} \qquad \qquad \tau_{L-par}^v(G) = (2+o(1))\cdot \tau_{par}^v(G),\]hold w.h.p.\ and in expectation.
	\end{thm}
	
	The proofs of the above theorems are based on a coupling between the Sequential and Parallel-IDLA processes. To construct this coupling we consider a (Parallel or Sequential) IDLA process on $G$ as an irregular 2-dimensional array $L$ where each element $L(i,j)\in V$. This array $L$ has $n$ rows representing the $n$ particles. Column $t$ represents time $t$, and thus $L(i,t)$ represents the vertex visited by walk $i$ at time $t$. We let $\rho_i$ denote the length of walk $i$, hence the index of each row $i$ goes from $0$ to $\rho_i$. We denote by $\mathcal{I}_L$ the set of all indices $(i,t)$ of the array $L$. 
	
	Given $(i,s), (j,t) \in \mathcal I_L$, we say that $(i,s)$ is smaller than $(j,t)$ in sequential order, written $(i,s) <_S (j,t)$ if either $(i<j)$ or $(i=j,s<t)$. Thus in sequential order, the block $L$ is read as \[L(1,0), L(1,1), \dots L(1,\rho_1),L(2,0),  \dots L(2,\rho_2),\dots , L(n,0), \dots ,L(n,\rho_n).\]
	Likewise we say that  $(i,s)$ is smaller than $(j,t)$ in parallel order, denoted by $(i,s) <_P (j,t)$ if either $(s<t)$ or $(s=t , i<j)$. So, in parallel order, the block $L$ is read as \[L(1,0), L(2,0),\dots, L(n,0),L(1,1),L(2,1), \dots, L(n,1), \dots ,L(1,r),L(2,r), \dots, L(n,r),\dots\] where if $r>\rho_i$ then $L(i,r)$ is empty so it is skipped.
	
	Note that if $L$ is a block representing a parallel or Sequential-IDLA the following property holds
	\begin{equation}\label{propA}
	\text{$L(i,\rho_i) \neq L(j,\rho_j)$ for each pair $i\neq j$.}
	\end{equation}
	If $L$ satisfies \eqref{propA} then $\{L(i,\rho_i): i \in [n]\} = V$ and the final element of each row is unique.

	A block $L$ satisfying \eqref{propA} represents a Sequential-IDLA process if and only if each row $i$ represents a path in $G$ from vertex $L(i,0)=v$ to $L(i,\rho_i)$ and for all $(i,t) \in \mathcal{I}_L$ \begin{equation}\label{seqdef}
	(i,t) \text{ is the first occurrence of vertex } L(i,t)\text{ in } L \text{ w.r.t.}<_S \text{ iff }t = \rho_i.
	\end{equation} This says that when $L$ is read in sequential-order the first time a new vertex is read it ends the current row. Similarly  a block $L$  satisfying \eqref{propA} is a realization of a Parallel-IDLA process if and only if  each row $i$ represents a path in $G$ from vertex $L(i,0)=v$ to $L(i,\rho_i)$ and and for all $(i,t) \in \mathcal I_L$ \begin{equation}\label{pardef}
	(i,t) \text{ is the first occurrence of vertex } L(i,t)\text{ in } L \text{ w.r.t.}<_P \text{ iff } t = \rho_i.
	\end{equation} 
	
	For a $2$-dim array $L$ we denote its total length (the work done) by $W(L)$, this is the total number of moves recorded by $L$ and thus $W(L):=\rho_1 + \cdots +\rho_n$.  Let $\operatorname{Seq}_v^m$, or $\operatorname{Par}_v^m$, denote the set of all sequential, respectively parallel, blocks representing realizations of IDLA starting from $v$ and total length $m$, i.e. $W(L)=m$. 
	
	To build the coupling between Sequential and Parallel-IDLA, we are going to use a series of ``Cut \& Paste'' transformations. Consider $(i,t) \in \mathcal{I}_L$, then define $\CP_{(i,t)}(L)$ as the block constructed by taking $L$ and cutting the cells $(i,t+1),\ldots, (i,\rho_i)$ and pasting it after the unique $(k,\rho_k)$ with $L(i,t) = L(k,\rho_k)$.

	\textbf{Example:} Represented below are $L$, a block on $V = \{1,2,3,4\}$, and $\CP_{(4,1)}(L)$ which is the result of applying the cut \& paste $\CP_{(4,1)}$ to $L$. 
	\begin{align*}	L&= \young(1,12,1223,121234)
	\qquad \qquad\quad  \CP_{(4,1)}(L)= \young(1,121234,1223,12) \end{align*}	
	While $\CP_{(1,0)}(L)= \CP_{(2,1)}(L) = \CP_{(3,3)}(L) = \CP_{(4,5)} = L$.
	Note that if $L$ satisfies property \eqref{propA}, then $L' = \CP_{(i,t)}(L)$ also satisfies \eqref{propA}. Property \eqref{propA} is an important invariant for our algorithms.
	\subsection{Algorithms}
	We propose two algorithms \textbf{StP} and \textbf{PtS}, formally specified by Algorithms \ref{StP} and \ref{PtS} below. The algorithm \textbf{StP} transforms a sequential process into a parallel and \textbf{PtS} transforms a parallel process into a sequential. The key component of both algorithms is the ``cut \& paste'' operation $\CP$.
	
	Both algorithms work as follows: a pointer moves through the input array $L$ in a fixed order and when the pointer sees a vertex label for the first time this label is added to the set $\mathcal{S}$ of seen vertices and a cut \& paste transform $\CP$ is applied to $L$ at this position before the pointer continues. The difference is that in \textbf{StP} the pointer explores columns then rows (i.e. in parallel order $<_P$), whereas \textbf{PtS} reads rows then columns (i.e. in sequential order $<_S$). 
	
	Broadly speaking the algorithms try to read the input array as if it was of the type specified by the output and if the input fails to have this form then it will edit it using the cut \& paste transform until it has the correct form.

	\medskip

	\noindent\makebox[\textwidth][c]{
		
		\begin{minipage}{\dimexpr.5\textwidth-.5\columnsep}
			\begin{algorithm}[H]\label{StP}
				\SetAlgoLined
				\KwResult{transforms a sequential array $L$ into a parallel array}
				
				$\mathcal S \leftarrow \emptyset$\;	$t \leftarrow 0$\;
				\While{$|\mathcal S| < n$}
				{
					
					\For{$i=1,\dots, n$}
					{
						
						\If{$(i,t) \in \mathcal I_{L}$ and $L(i,t) \not\in \mathcal S$}
						{
							$\mathcal S \leftarrow \mathcal S \cup \{L(i,t)\}$\;
							$L \leftarrow \CP_{(i,t)}(L)$\;
							
						}
						
					} 	$t \leftarrow t+1$\;
				}
				return $L$\; 
				\vspace{11.5pt}
				\caption{Sequential to Parallel ($\mathbf{StP}$)\!\!\!\!\!\!}
			\end{algorithm}
			\hfill
		\end{minipage}
		\begin{minipage}{\dimexpr.5\textwidth-.5\columnsep}
			\begin{algorithm}[H]\label{PtS}
				\SetAlgoLined
				\KwResult{transforms a parallel array $L$ into a sequential array}
				
				$\mathcal S \leftarrow \emptyset$\;
				\nl \For{$i = 1,\dots,n$}
				{
					$t \leftarrow 0$\;
					\nl\While{$(i,t) \in \mathcal I_L$}
					{
						\nl\If{$L(i,t) \not\in \mathcal S$}
						{
							\nl	$\mathcal S \leftarrow \mathcal S \cup \{L(i,t)\}$\;
							\nl$L \leftarrow \CP_{(i,t)}(L)$\;
							$\mathbf{exit(while)}$

						}	$t \leftarrow t+1$\;
						
					}
				}
				return $L$\;
				
				\caption{Parallel to Sequential ($\mathbf{PtS} $)\!\!\!\!}
			\end{algorithm}
	\end{minipage}}

	\medskip
	
	The set $\mathcal{S}=\mathcal{S}(L,k)$ stores the different values of $L(i,j)$ observed after $k$ iterations of the innermost loop. The algorithms terminate once they have scanned the whole array, this is the first time when $|\mathcal S| = n$. Sometimes they may apply $\CP_{(i,j)}$ with $j=\rho_i$, this leaves $L$ unchanged.

	\begin{lem}[Correctness and bijectivity of Algorithms \ref{StP} \& \ref{PtS}]\label{cor1} The following holds,
		\begin{itemize}
			\item $\pts$ is a bijection from $\operatorname{Par}_v^m$ to $\operatorname{Seq}_v^m$.
			\item $\stp$ is a bijection from $\operatorname{Seq}_v^m$ to $\operatorname{Par}_v^m$.
		\end{itemize}
	\end{lem}
	\begin{proof}Observe that during the running of the $\textbf{PtS}$ and $\textbf{StP}$, Algorithms \ref{StP} \& \ref{PtS}, the only changes made to the input array $L$ are a sequence of cut \& paste transforms $\CP_{i_1,t_1},\CP_{i_2,t_2}\dots $. Since each cut \& paste transform preserves Property \eqref{propA} it follows that $\textbf{PtS}$ and $\textbf{StP}$ preserve \eqref{propA}. Likewise cutting \& pasting preserves total length, thus so do $\textbf{PtS}$ and $\textbf{StP}$. Recall that the operator $\CP_{(i,t)}$ cuts and pastes the random walk trajectory $(i,t+1), \ldots, (i,\rho_i)$ onto the unique $(k,\rho_k)$ with $L(i,t) = L(k,\rho_k)$. Thus row $k$ in $L' = \CP_{(i,t)}$ is a valid path from vertex $L(0,k)$ to $L(i,\rho_i)$. 	
		
		For $\textbf{PtS}$ we must check that if $L \in \operatorname{Par}_v^m$, then $\textbf{PtS}(L) \in \operatorname{Seq}_v^m$, i.e. $\pts(L)$ satisfies \eqref{seqdef}. Recall that the $\pts$ algorithm reads the input array $L$ in sequential order and when a vertex label is seen for the first time at some position $(i,j)$ it applies the cut \& paste transform $\CP_{(i,j)}$ and the pointer moves to the next row.	If $(i,j+1)$ is non-empty then $\CP_{(i,j)}$ pastes the remainder of row $i$ to some row $i'$ with endpoint value $L(i,j)$. Observe that $i'>i$ since $(i,j)$ is the first occurrence of $L(i,j)$ in sequential order. Thus each new vertex found w.r.t. $ <_S$ forms an endpoint as it is cut when it is first discovered and nothing else can be pasted onto that row later by the algorithm. This proves that $\pts(L)$ is a valid Sequential-IDLA block.

		Likewise for $\stp$ let $L \in \operatorname{Seq}_v^m$ and we check $\stp\left( L\right)$ satisfies \eqref{pardef}. Suppose when reading $L$ in parallel order $(i,j)$ is the first occurrence of $L(i,j)$, $\stp$ will apply $\CP_{(i,j)}$ and continue to read the array in parallel order. Position $(i,j)$ is now fixed as the end point of row $i$ as no later copy \& paste can alter this row. This holds since to paste something else onto row $i$ we would have to see vertex $L(i,j)$ for the first time (again) later in parallel order which cannot happen.
		
		For injectivity let $\mathbf{F}$ represent either of the maps $\pts,\stp$, and $L ,L'$ be distinct arrays both from $\operatorname{Par}_v^m$ or $\operatorname{Seq}_v^m$ respectively. Assume for a contradiction that $\mathbf{F}(L)= \mathbf{F}(L')$. Since $L\neq L'$ there is a first position $(i,j)$ at which they differ  w.r.t. $<_S$ or $<_P$, i.e. $L(i,j) \neq L'(i,j)$. It cannot be the case that $L(i,j)= \emptyset$ and $L' \neq \emptyset$, or vice versa, since otherwise the arrays must differ at position $(i,j-1)$ which occurs before $(i,j)$ in either ordering. Let $(i,j)$ be the current position when $\mathbf{F}$ is is running on $L$ and $L'$. If $L(i,j) \not\in \mathcal{S}(t,L)$ and $L'(i,j) \not\in \mathcal{S}(t,L')$ then $\CP_{(i,j)}$ is applied and the position $(i,j)$ is now fixed in both arrays, i.e. $\mathbf{F}(L)(i,j) \neq \mathbf{F}(L')(i,j)$, a contradiction. Similarly if $L(i,j) \in \mathcal{S}(t,L)$ and $L'(i,j) \in \mathcal{S}(t,L')$ then no transform is applied and the positions are fixed. Otherwise the element at $(i,j)$ is seen in one array and not in the other, i.e. $L(i,j) \not\in \mathcal{S}(t,L)$ and $L'(i,j) \in \mathcal{S}(t,L)$. This is a contradiction as $(i,j)$ is the first position at which $L$ and $L'$ differ.
		
		For bijectivity since $\stp:\operatorname{Seq}_v^m\rightarrow \operatorname{Par}_v^m$ and $\pts:\operatorname{Par}_v^m\rightarrow \operatorname{Seq}_v^m$ are both injections and $\operatorname{Seq}_v^m,\operatorname{Par}_v^m$ are finite it follows that $|\operatorname{Seq}_v^m|= |\operatorname{Par}_v^m|$. Thus $\stp,\pts$ are surjections.\end{proof}
	\begin{rem}
		One can prove $\stp$ has inverse $\pts$, we omit the proof as we do not use this fact.  
	\end{rem}

	\begin{lem}\label{ordering}
		Let $L \in \operatorname{Seq}_v^m$. Then $\max_{i\in \mathcal{I}_{L}} \rho_i \leq \max_{i\in \mathcal{I}_{\stp(L)}} \rho_i $.  
	\end{lem}
	\begin{proof} Assume for a contradiction that $\max_{i\in \mathcal{I}_{L}} \rho_i > \max_{i\in \mathcal{I}_{\stp(L)}} \rho_i $. This means that each row attaining maximum length in $L$ must have a section cut and pasted to a row of shorter length by the \textbf{StP} algorithm. However the \textbf{StP} algorithm runs in parallel order and cannot paste onto a cell which it has already read. Thus any row suitable to receive the end of the current row must have its end point in the same column or a column to the right of the current one. This cannot decrease the length of the longest row.
	\end{proof}

	We now have what we need to prove that $\tau_{seq}^v(G) \preceq \tau_{par}^v(G) $ for any $G$ and $v \in V(G)$. 	
	\begin{proof}[Proof of Theorem \ref{PStocS}]
		By Lemma \ref{cor1} $\stp$ is a bijection between $\operatorname{Par}_v^m$ and $\operatorname{Seq}_v^m$. Thus we can pair every sequential process $L$ of total length $W(L)=m$ with a unique parallel process $L'$ of total length $W(L')=m$. Both $L$ and $L'$ visit the same vertices with the same frequency and in the same order, thus the probability of each vertex sequence of total length $m$ in either process is the same. This implies that the total lengths of the processes are distributed identically. 
		
		Lemma \ref{ordering} states that for this pair the longest row in $L'$ is at least as long as the longest row in $L$. Thus for any $k,m\geq 0$, \[\Pro{ \max\limits_{i\in \mathcal{I}_{L}} \rho_i \geq k \,\Big|\, W(L) =m } \leq \Pro{ \max\limits_{i\in \mathcal{I}_{L'}} \rho_i\geq k \,\Big|\, W(L') =m   }.\] This implies the result since $\tau_{seq}^v(G)$ and $\tau_{par}^v(G)$ are given by the length of the longest row in the sequential and parallel processes respectively.
	\end{proof}
	
	In the other direction we now prove $\Ex{\tau_{par}^v} = \BO{ \Ex{\tau_{seq}^v}\cdot\log n}$ for any $G$ and $v \in V$. 	
	\begin{proof}[Proof of Theorem \ref{couplog}.]
		Let $L$ be a Parallel-IDLA block and $\sigma$ be a random permutation of $\{2,\ldots, n\}$. Let $\sigma(L)$ be the block that results from permuting the rows of $L$ using $\sigma$. The block $\sigma(L)$ represents a Parallel-IDLA process where conflicts between particles are solved by giving priority to particles with least value of $\sigma(index)$ (instead of least $index$, as per the definition of Parallel-IDLA).  Also, for simplicity we fix $\sigma(1) = 1$.  Note that $L$ and $\sigma(L)$ have the same rows, and thus the maximum row-length is the same in both blocks. We remark that $\pts$, Algorithm \ref{PtS}, still produces a valid sequential array even if the input is $\sigma(L)$ instead of $L$. 
		
		Let $L$ be an arbitrary parallel array and consider a run of $\pts$, Algorithm \ref{PtS}, on $\sigma(L)$ where we do not reveal $\sigma$ in advance. Instead we reveal the permutation $\sigma$ row by row as $\pts$ reads the array in sequential order (in other words, instead of running $\pts(\sigma(L))$, we equivalently run $\pts(L)$ but we read rows in random order, starting with row $1$ $(=\sigma(1))$ of $L$, and then rows $\sigma(2), \sigma(3), \ldots, \sigma(n)$. This is equivalent to replacing $i$ by $\sigma(i)$ in lines 1-5 of Algorithm \ref{PtS}). Note that the Cut \& Paste operation is unaffected by not revealing the order of the rows. This holds because the Cut \& Paste transform only pastes behind unread rows, independent of their location in the array $L$ and what is more, there is only one row where we can paste a cut section by property \eqref{propA}. Consider the largest row (or choose one arbitrarily if there is more than one) in the original block $L$. We shall paint this row red and call the last cell $\xi$. During the running of $\pts(L)$ the marked cell $\xi$ moves from row to row because of the Cut \& Paste operations. Here is the key observation: If $\ell$ is the length of the original red row and $\xi$ moves no more than $N$ times then in the output array $\pts(L)$ has a row of length at least $\ell/N$. This holds because the red row was partitioned $N$ times and thus one of the pieces has to have length at least $\ell/N$ 
		
		Let $i_k$ be the iteration (how many rows we have read) by the $k^{th}$ time $\pts$ reads a row containing the marked cell $\xi$. When we read a row which contains $\xi$ for first time in iteration $i_1$, we may apply a Cut \& Paste somewhere in this row (if not we are done). If so $\xi$ would find itself at the end of an unread row $x_2$ of $L$, which will be read in a (random) iteration $i_2$, i.e. $\sigma(i_2) = x_2$. Note $i_2$ is a uniform random value in $\{i_1+1,\ldots, n\}$. In iteration $i_2$, we read the row with the marked cell and again, the algorithm might cut and paste this row behind an unread row $x_3$ which will be read at some time $i_3$, which is again uniformly random in $ \{i_2+1,\ldots, n\}$, and so on. Each time we make a cut and paste the index $i_j$ of the recipient row will be in the latter half of the list $\{i_j+1,\dots ,n\} $ with probability $1/2$. Thus since $\pts$ works through this list in order the expected length of the list of possible positions for the next value $i_{j+1}$ halves every iteration. We cannot keep halving this list indefinitely because either at some point a row ended by $\xi$ is not cut or $\xi$ is in the last row to be read (which is never cut). Thus the number of times $\xi$ moves (i.e. expected times the longest row is cut) is at most $C\log n$ with probability at least 1/2 by Markov's inequality. Denote by $X$ the (random) number of times we cut a row containing the marked cell $\xi$. Let $\ell$ be the length of the longest row of $L$, and $\ell'$ the random variable representing the length of the longest row of $\pts(L)$ using a random permutation $\sigma$.  Conditional on cutting $L$'s longest row $X$ times, we have must have at least one row of length $\ell/X$ once the algorithm has terminated. Thus, given the block $L$ with largest row $\ell$, we have
		$$\Ex{\ell'\,|\,L} > \Ex{\ell'\,|\,L,X\leq C \log n}\cdot \frac{1}{2}\geq \frac{\ell}{2C\log n}.$$
		By taking expectation over all blocks $L$ generated from a Parallel-IDLA with a random $\sigma$ we conclude the result.
	\end{proof}

	\subsection{Uniform-IDLA}
	
	Recall that in the Sequential-IDLA we run the walks one by one in order and walk $i+1$ starts only after walk $i$ has settled, while in the Parallel-IDLA all particles walk simultaneously until they settle, breaking ties by settling the particle with smallest index. In either Sequential or Parallel we are interested in the longest walk. Another natural way to run the IDLA process is in uniform order: we choose a random unsettled particle and move it to a random neighbouring vertex which it settles on if unoccupied. We call this process the Uniform-IDLA. This process can be seen as lying between the Sequential and Parallel-IDLA models. To sample from the Uniform-IDLA process, we first consider an infinite sequence $R = (R_i)$ where the $R_i$s are independent random variables sampled from $\{2,\ldots, n\}$.
	Then we run the Uniform-IDLA as follows: First particle $1$ settles at the origin, so the origin is occupied. Then, at each time-step $t \geq 1$, if particle $R_t$ is unsettled, it moves to a random neighbour, otherwise it stays in its current location. If such neighbour is not occupied, particle $R_t$ settles on it and the vertex is now occupied.
	
	Clearly for some sequences $R$ the process may never terminate, for example $R=(1,1,1,\dots )$. We say that an sequence $R$ on $n$ indices is index-repeating if for any index $i\in \{1,\dots,n \}$ and any $T\geq 0$, there exits some $t\geq T$ such that $R_t=i$. 
	
	\begin{rem}\label{a.s.index}For any fixed $n$ and any fixed distribution $\mathcal{D}$ on the indices with full support, the $R$ obtained by sampling indices according to $\mathcal{D}$ will be index-repeating almost surely.
	\end{rem}
	
	Given an index-repeating ordering $R$, we can find a bijection between the Uniform-IDLA and Parallel-IDLA. An $R$-block is defined in the same fashion as a parallel block, i.e. $L(i,j)$ represents the position of the $i$-th particle after $j$ jumps, but additionally, we associate to every $(i,j) \in \mathcal I_L$ an integer $T(i,j)$. This $T$ is called the timing array and defined as $T(i,j) = t$ if $R_t = i$ for $j$-th time and $T(i,0) = 0$ for all particles $i$. Note that using the block and timing array we can reconstruct the uniform process as we have not only the paths but the time-steps when the particles moved. Whenever we speak of an $R$-block we shall assume that $R$ is index repeating.
	
	The bijection between an $R$-block and a parallel block is defined algorithmically in the same fashion as before. To transform an $R$-block into a parallel block we just apply $\stp$, Algorithm \ref{StP}, to the $R$-block oblivious to $R$ since $\stp$ reads in parallel order. However to transform a parallel block into an $R$-block, we must read the block in the order given by $T(i,t)$ (i.e. read the block with smallest value $T(i,t)$, then the second smallest, etc..) and apply $\CP_{i,j}$ whenever the vertex $L(i,j)$ is read for first time. It is very important that now when applying the Cut \& Paste operation we move not only the cells containing a portion of the path but also the times $T(i,t)$ associated to those cells, i.e. if cell $(i,t)$ moves to $(j,s)$ then $T(j,s)$ gets the value of $T(i,j)$, while $T(i,t)$ is left undefined. Pseudo-code for the procedure we have just described is given in Algorithm \ref{PtU}.
	
	\begin{algorithm}\label{PtU}
		\SetAlgoLined
		\KwResult{transforms a parallel array $L$ and order sequence $R$ into a $R$-Uniform array}
		
		$\mathcal S \leftarrow \emptyset$\;
		$C\leftarrow \text{list of cells $(i,j)$ ordered  by  $T(i,j)$ in increasing order}$\;
		$k \leftarrow 0$\;
		\While{$|\mathcal S|<n$}
		{	
			$k \leftarrow k+1$\;
			\nl	$(i,j) \leftarrow C(k)$\;
			\nl	\If{$L(i,j) \not\in \mathcal S$}
			{
				$\mathcal S \leftarrow \mathcal S \cup \{L(i,j)\}$\;
				$L \leftarrow \CP_{(i,j)}(L)$\;

			}
		}
		return $L$\;
		
		\caption{Parallel to $R$-Uniform ($\mathbf{PtU}_R $)\!\!\!\!}
	\end{algorithm}

	Let $\operatorname{Unif}_{R,v}^m$ be the set of all Uniform-IDLA blocks with ordering $R$ starting from $v$ with total number of steps $m$. Then, using similar arguments to the sequential-parallel case we obtain.
	\begin{thm}\label{thm:UtP}
		For any fixed index-repeating sequence $R$ on $n$ indices, there is a bijection between $\operatorname{Unif}_{R,v}^m$ and $\operatorname{Par}_v^m$. Moreover the number of steps taken by the longest walk of the Uniform-IDLA is stochastically dominated by the number of steps in the longest walk of the Parallel-IDLA.   
	\end{thm}
	\begin{proof}
		The bijection follows from injectivity and correctness of $\stp$ and $\mathbf{PtU}_R$ (as in Theorem \ref{cor1}). Then as in the proof of Theorem \ref{PStocS} we run $\stp$ and apply Lemma \ref{ordering}. This Lemma still applies as $\stp$ is oblivious to the ordering of the input array.  
	\end{proof}
	Observe however that the dispersion time of the Uniform array is not determined purely by the number of steps/length of the longest row but by the values $T(i,j)$ of the timing array.

	\subsection{Continuous-Time IDLA}
	
	In this section we consider continuous-time versions of the Sequential and Uniform-IDLA process. By this we mean running these IDLA processes with random walks with exponential rate $1$ jumps. We shall need the following concentration result.   
	\begin{lem}\label{Erlconc}Let $X$ be an $\Erl{n}{\lambda} $ random variable and $Y=\sum_{i=1}^{n}Y_i$ where $Y_i$ are independent $\geo{p} $ random variables. Then for any $\delta >0$ and $0<\varepsilon<1 $ ,
		\begin{enumerate}[(i)]
			\item \label{itmerl} $\displaystyle \Pro{X \geq \frac{(1+\delta)n}{\lambda}  } \leq e^{-\frac{\delta n}{2}}\;\;\;\;\,\qquad$ and $\qquad \displaystyle \Pro{X \leq \frac{(1-\varepsilon)n}{\lambda}  } \leq e^{-\frac{\varepsilon n}{2}}.$ 
			\item \label{itmsumgeo} $\displaystyle \Pro{Y\geq \frac{(1+\delta)n}{p}} \leq e^{-\frac{\delta^2n}{2(1+\delta)}} \qquad $ and $\qquad \displaystyle \displaystyle \Pro{Y\leq \frac{(1-\varepsilon)n}{p}}  \leq e^{-\frac{\varepsilon^2 n}{2(1-2\varepsilon/3) }}.$ 
		\end{enumerate} 
	\end{lem}
	\begin{proof}\textit{Item \eqref{itmerl}}: If $X$ is  $\Erl{n}{\lambda} $ then $ \Ex{X} = n/\lambda$  and $ \Ex{e^{tX}}= \left(1-t/ \lambda \right)^{-n}$ for all $t<\lambda$. Now by Markov's inequality for any $t<\lambda$,
		\[ \Pro{X\geq (1+\delta) \mu } \leq e^{-t(1+\delta) \mu }\Ex{e^{tX}}\leq e^{-t(1+\delta)n/\lambda }\left(1-t/ \lambda \right)^{-n}\leq e^{-t(1+\delta)n/\lambda } e^{tn/\lambda } \leq e^{-\delta tn/\lambda }. \]  
		By considering $-X $ one can also show $\Pro{X\leq  (1-\varepsilon) \mu }\leq e^{-\varepsilon tn/\lambda } $ provided $\varepsilon<1 $. Since $t <\lambda  $ was arbitrary the result follows by choosing $t=\lambda/2$.
		
		\noindent\textit{Item \eqref{itmsumgeo}}: Following \cite{GeoIneq}, if $Y\geq k\Ex{Y}$ then we have less than $n$ successes in $k\Ex{Y}$ Bernoulli trials with success probability $p$. Thus $\Pro{Y\geq (1+\delta)n/p}\leq\Pro{\bin{(1+\delta)n/p}{p}< n}$ and  \[\Pro{\bin{(1+\delta)n/p}{p}< n} = \Pro{\bin{(1+\delta)n/p}{p}< (1+\delta)n - \delta n} \leq e^{-\frac{\delta^2n}{2(1+\delta)} } , \]by \cite[Thm. 3.2]{ChernSurvey}. Similarly for the lower bound \[\Pro{Y\leq  (1-\varepsilon )n/p}\leq\Pro{\bin{(1-\varepsilon)n/p}{p}\geq  n}\leq e^{-\frac{\varepsilon^2n^2}{2((1-\varepsilon)n+\varepsilon n/3 ) }}  = e^{-\frac{\varepsilon^2 n}{2(1-2\varepsilon/3) }}. \] 
	\end{proof}
	
	For the Sequential-IDLA it is easy to consider its continuous-time analogue, the ContSeq-IDLA, we just have random walks that jump at times given by a Poisson process of intensity $1$. Also, we can easily sample from the ContSeq-IDLA by sampling a standard (discrete time) Sequential-IDLA and then considering independent exponential times of mean $1$ between the jumps. Let $\tau_{c-seq}^v(G)$ be the time it took to the slowest particle to settle in the ContSeq-IDLA from $v \in V$. 
	\begin{thm}\label{c-seq} Let $G$ and $v\in V(G)$. Then
		\begin{enumerate}[(i)]
			\item \label{itm:co1}$\displaystyle \tau_{c-seq}^v(G) =\BT{\tau_{seq}^v(G)}$ holds w.h.p.\ and in expectation.  
		\end{enumerate}If in addition $\tau_{seq}^v(G)= \lohm{\log n}$ then	
		\begin{enumerate}[(i)]\setcounter{enumi}{1}
			\item \label{itm:co2}  $\displaystyle{\tau_{c-seq}^v(G) = (1+o(1))\cdot \tau_{seq}^v(G)}$ holds w.h.p.\ and in expectation.\end{enumerate}
	\end{thm}
	\begin{proof}We sample a ContSeq-IDLA by sampling $L$, a Sequential-IDLA, and a $\expy{1}$ random variable for each walk step in $L$. Thus conditional on $L$ walk $i$ in the ContSeq-IDLA has length $\Erl{\tau_i}{1}$. Let $\mathcal{E_\ell}$ be the event that $L$ contains at least one row of length at least $\ell\geq 1 $.
		
		For the upper bound conditional on $\rho_*$ being length of the longest row of our sampled Sequential-IDLA we can stochastically dominate the length of any walk in the ContSeq-IDLA by an independent $\Erl{\rho_*}{1}$ random variable. Observe that $\Pro{\tau_{c-seq}^v>\rho_* + \delta \rho_*\mid  \rho_* } \leq   n\cdot  e^{-\delta \rho_*/2}$ by Lemma \ref{Erlconc} \eqref{itmerl}. Thus conditional on $\mathcal{E_\ell}$ we can take $\delta =(4\log n)/\ell$ so w.p. $1-\lo{1/n}$ at most $4\log n$ extra steps are taken by any continuous walk. This gives \begin{equation}\label{upwhp}\Pro{\tau_{c-seq}^v\leq (1+4\log( n) /\ell)\tau_{seq}^v} \geq 1 - \Pro{(\mathcal{E}_\ell)^c}-\lo{1/n}. \end{equation} Let $\rho_*\geq 1$ be the length of the longest row of $L$, and observe that $\Pro{\tau_{c-seq}^v>\rho_* + i\log n \mid \rho_*}\leq ne^{-(i\log n)/2} $ follows by taking $\delta =(i\log n)/\rho_* $  in Lemma \ref{Erlconc} \eqref{itmerl}. Thus   \begin{align}
		\Ex{\tau_{c-seq}^v \mid \rho_*  }&\leq \rho_* +  2\log n + (\log n)\cdot \sum_{i=2}^{\infty}ne^{-(i\log n)/2} =\rho_*+\BO{\log n }.\label{eqn:expbdd2}
		\end{align} Observe that $\Ex{\rho_*} =  \Ex{\tau_{seq}^v }$ as $\rho_* $ is the longest row of $L$, thus $\Ex{\tau_{c-seq}^v }\leq \Ex{\tau_{seq}^v}+\BO{\log n }$.
		Note by definition that if $\Pro{\mathcal{E}_\ell}=1-\lo{1}$ then $\Ex{\tau_{seq}^v}\geq (1-\lo{1})\ell$. Lemma~\ref{lemma:whpLB} states that for any graph and any start vertex,  $\Pro{\mathcal{E}_{(\log n )/14 } }=1-\lo{1} $. Thus the upper bound in \eqref{itm:co1} holds in expectation. The upper bound in \eqref{itm:co2} follows similarly assuming $\Pro{\mathcal{E}_{\omega(\log n )} }=1-\lo{1} $. By \eqref{upwhp} the upper bound in cases \eqref{itm:co1} \& \eqref{itm:co2} also hold w.h.p.. 
		
		For the lower bounds take one walk from $L$ of the maximum length $ \rho_*$ and consider its length in the ContSeq-IDLA process. This has Gamma distribution $\Erl{\rho_*}{1}$ and thus  \[\Pro{\tau_{c-seq}^v< \rho_{*} - \sqrt{\ell}\; \Big|\; \rho_{*}}\cdot \mathbf{1}_{\mathcal{E}_\ell} \leq  e^{-\sqrt{\ell}/2}\] by Lemma \ref{Erlconc} \eqref{itmerl}. The w.h.p.\ lower bounds for \eqref{itm:co1} and \eqref{itm:co2} follow by taking expectations of the above, since in both cases $\ell =\Omega( \log n )$ and $\Pro{ \left(\mathcal{E}_\ell\right)^c} = \lo{1}$. This holds in expectation also.
	\end{proof}

	It is natural to consider the continuous-time version of the Uniform-IDLA, we call this the CTU-IDLA. In this process each particle has an exponential clock with rate $1$. Then, as long as the particle is not settled, when the clock rings the particle moves to a random neighbour and settles if possible. Note that this is equivalent to running the discrete-time Uniform-IDLA with a uniformly random sequence $R$ but waiting an amount of time distributed $\expy{1/(n-1)}$ between each step in $R$ (recall particle $1$ occupies the origin and $R_t$ takes values in $\{2,\ldots, n\}$). Alternatively, we can sample the CTU-IDLA by using $\mathbf{PtU}_R $, Algorithm \ref{PtU}. First, sample a $L$, a (discrete-time) Parallel-IDLA, then run Algorithm \ref{PtU} but using a list $C$ built from a timing array $T$ populated as follows: Set $T(i,0)=0$ for each $i$. Then let $T(i,j+1)= T(i,j)+X_{i,j}$ where $\{X_{i,j}\}_{i \in [n],j\in \mathbb{N}}$ are independent $\expy{1}$ random variables. We shall name this procedure $\mathbf{PtU}_C$. This procedure can also be seen as running Algorithm \ref{PtU} but instead of using the list $C$ to choose the next cell $(i,t)$ (line 1), each row of the block has a exponential clock of mean $1$. When the clock of row $i$ rings, the algorithm chooses the first unread cell of row $i$ (if it exists), and proceeds with line $2$. One can show this algorithm is (almost surely) correct due the bijection between $\operatorname{Unif}_{R,v}^m$ and $\operatorname{Par}_v^m$ for a fixed ordering $R$ established in Theorem~\ref{thm:UtP} and Remark \ref{a.s.index}. Let $\tau_{c-unif}^v$ be the time it takes the CTU-IDLA started from $v$ to settle all particles. 
	
	\begin{thm}\label{conttimepar}  Let $G$ be a connected graph and $v\in V(G)$. Then
		\begin{enumerate}[(i)]
			\item \label{itm:cu1}$\displaystyle{\tau_{c-unif}^v(G) =\BT{\tau_{par}^v(G)}}$ holds w.h.p.\ and in expectation.  
		\end{enumerate}If in addition $\tau_{par}^v(G)= \lohm{\log n}$ then	
		\begin{enumerate}[(i)]\setcounter{enumi}{1}
			\item \label{itm:cu2}  $\displaystyle{\tau_{c-unif}^v(G) = (1+o(1))\cdot \tau_{par}^v(G)}$ holds w.h.p.\ and in expectation.\end{enumerate}
	\end{thm}
	\begin{proof}
		We can sample a CTU-IDLA as described above by sampling $L$, a Parallel-IDLA, and running $\ptu_C$ on $L$. By Theorem \eqref{thm:UtP} the longest row of $\mathbf{PtU}_C(L)$ is no longer than the longest row of $L$, thus conditional on $L$ walk $i$ in the CTU-IDLA has length $\Erl{\tau_i}{1}$. The proofs of the upper bounds in \eqref{itm:cu1} and \eqref{itm:cu2} now follow by the exact same argument as in the proof of Theorem \ref{c-seq}. Let $\mathcal{E_\ell}$ be the event that $L$ contains at least one row of length at least $\ell\geq 1 $.

		For the lower bound take one row from $L$ of the maximum length $\rho_*$ (assume the label of this row is $i$) and consider the action of $\mathbf{PtU}_C$ on the cells in this row. If no cut is made during the running of $\mathbf{PtU}_C(L)$ then the length of $i$ stochastically dominates an $\Erl{\rho_*}{1}  $ random variable. Suppose that a Cut \& Paste transform is applied to $i$ at a cell containing vertex $v$ and the remainder of this row is pasted onto row $j$. Although row $j$ may have contained less cells before $v$ than the number of steps taken by row $i$ to reach $v$, the amount of time (with respect to the clock) it takes particle $j$ to reach $v$ must be at least as long as the time for particle $i$ to reach $v$ (otherwise the Cut \& Paste would not have been applied). Thus conditional on this Cut \& Paste the length of row $j$ stochastically dominates an $\Erl{\rho_*}{1}$ random variable. Now, as in Theorem \ref{c-seq}, \[\Pro{\tau_{c-unif}^v< \tau_{par}^v - \sqrt{\ell}\; \Big|\; \mathcal{E}_\ell} = e^{-\sqrt{\ell}/2}\] by Lemma \ref{Erlconc} \eqref{itmerl}. The lower bounds for \eqref{itm:cu1} and \eqref{itm:cu2} follow since in both cases $\ell =\Omega( \log n )$.  
	\end{proof}
	
	\subsection{Lazy IDLA}
	
	Consider the lazy versions of the discrete-time Sequential and Parallel-IDLA models, where with probability $1/2$ particles stay put and otherwise choose a neighbour uniformly. Note that all our previous results using the coupling via the block representation are also valid for lazy walks as for example one can simply consider the graph with the addition of (multi)-loops at each vertex. Indeed, they are valid for any block that is generated by using a Markov chain to move the particles.
	
	Let $\tau_{L-seq}^v(G)$, $\tau_{L-par}^v(G)$, be the number of steps needed to complete the lazy Sequential, respectively lazy Parallel, IDLA process started from $v$. 
	
	Although we are mainly concerned with the simple random walk IDLA models would like to be able to switch to the lazy setting at times as it allows us to use mixing time results. For the Sequential it is fairly clear that up to lower order terms the lazy sequential is a factor of $2$ slower than the Parallel, using the continuous time Uniform-IDLA we can also show this for Parallel-IDLA.

	\begin{proof}[Proof of Theorem \ref{lazy2normal}]We begin by proving the results for the sequential processes.
		
		Let $\mathcal{E_\ell}$ be the event that $L$ contains at least one row of length at least $\ell\geq 1 $.  We sample a lazy Sequential-IDLA by coupling with a simple Sequential-IDLA $L$ and adding in lazy steps w.p. $1/2$ for each walk. Thus each row $i$ with length $\tau_i$ in the sequential process has length $\sum_{i=1}^{\tau_i}Y_i$ in the coupled lazy process, where the $Y_i$ are independent $\geo{1/2}$ random variables. For the upper bound let $\rho_*$ be the length of $L$'s longest row, then we have $\Pro{\tau_{L-seq}^v>2\rho_* + k2\rho_*\mid  \rho_* } \leq   n\cdot  e^{-\frac{k^2 \rho_*}{2(1+k)}}$ by Lemma \ref{Erlconc} \eqref{itmsumgeo}. Similarly to Theorems \ref{c-seq} the w.h.p.\ upper bounds for $\tau_{L-seq}^v$ follow by conditioning on $\mathcal{E}_{\ell}$ for the two cases of $\ell$ in the statement.
		
		For upper bounds in expectation if we condition on $\left(\mathcal{E}_{\ell}\right)^c$ then $ 1\leq \rho_* \leq \ell $, it follow that 
		\begin{align*}
		\Ex{\tau_{L-seq}^v\mid (\mathcal{E}_\ell)^c}&< 2\ell + 6\ell \log n + (2 \ell \log n)\cdot\sum_{i=3 }^{\infty} \Pro{\tau_{L-seq}^v>\ell  + 2i\ell \log n }\\ &\leq \BO{\ell\log n } + (2\ell \log n )\cdot \sum\limits_{i=3}^{\infty }e^{-i^2/2(i+1)}=\BO{\ell\log n }.
		\end{align*}
		If $\ell\geq \log n /14$ then by Lemma~\ref{Erlconc} \eqref{itmsumgeo} \[ \Pro{\tau_{L-seq}^v>2\rho_* + \sqrt{i\rho_*\log n} \mid \rho_* }\mathbf{1}_{\mathcal{E}_{\ell}}\leq n\cdot e^{-\frac{i\log n}{2\left((i\log n)/\rho_* +1\right)}}\leq n\cdot e^{-(\sqrt{i}\log n)/10},\] thus $\Ex{\tau_{L-seq}^v\cdot \mathbf{1}_{\mathcal{E}_{\ell}} \mid \rho_*  } =2\rho_*+\BO{\sqrt{\rho_*\cdot \log n}},$ similar to \eqref{eqn:expbdd2}. By Jensen's (concave) inequality
		\begin{align} \label{exuplazy} \Ex{\tau_{c-seq}^v }&\leq 2\Ex{\tau_{seq}^v} +\BO{\sqrt{\Ex{\tau_{seq}^v}\cdot \log n} } + \BO{\ell\cdot \log n }\cdot \Pro{\left(\mathcal{E}_\ell\right)^c}.\end{align}
		Thus for any graph $G$ and $v\in V$, $ \Ex{\tau_{L-seq}^v }=\BO{\Ex{\tau_{seq}^v}} $ by Lemma \ref{lemma:whpLB}. If $\Pro{\left(\mathcal{E}_{\ell}\right)^c}\leq 1/\ell $ for some $ \ell = \lohm{\log n }$ then $\Ex{\tau_{c-seq}^v }\leq (2+\lo{1})\Ex{\tau_{seq}^v} $ by \eqref{exuplazy}.

		For the w.h.p.\ lower bound, conditional in $i$ being a walk of maximum length $\rho_*$ in L,  walk $i$ in the L-Seq-IDLA has length $\sum_{i=1}^{\rho_*} Y_{i}$, where $Y_i\sim \geo{1/2}$. So by applying Lemma \ref{Erlconc} \eqref{itmsumgeo},   \[\Pro{\tau_{L-seq}^v< 2\rho_* - (\log n)^{1/3}\sqrt{\rho_*}\; \Big|\; \rho_*}\cdot \mathbf{1}_{\mathcal{E}_\ell} \leq e^{-(\log n)^{1/3}/2(1-o(1))} = \lo{1} .\] Thus, since $\Pro{ \left(\mathcal{E}_\ell\right)^c} = \lo{1}$ for any $G,v\in V$ by Lemma \ref{lemma:whpLB}, talking expectations of the equations above yields $\tau_{L-seq}^v\geq (2-\lo{1})\tau_{seq}^v$ w.h.p\ . Thus this also holds in expectation.
		
		We now prove the bounds for the Parallel processes, the proof technique will be slightly different. First assume that for any $G$, $v\in V$ and some $\ell=\omega(\log n )$ we have $\Pro{ \left(\mathcal{E}_\ell\right)^c} = \lo{1}$.  In this case we know that $\tau_{par}^v(G) = (1+o(1))\tau_{c-unif}^v(G)$ w.h.p.\ and in expectation from Theorem \ref{conttimepar}. Consider the CTU-IDLA but using clocks of mean $2$ and use $\tau_{2-c-unif}^v(G)$ to denote the dispersion time of this process. It is clear that we can couple the clocks of mean $1$ and $2$ to give $\tau_{2-c-unif}^v(G) = (2+o(1))\tau_{c-unif}^v(G)$ w.h.p.\ and in expectation. Note that sampling from this process is equivalent to sampling from the Uniform-IDLA of mean $1$, but ignoring the ring of the clock with probability $1/2 $ (Poisson thinning). Consider the graph $\tilde{G}$, this is $G$ but to each vertex we add as many self loops as neighbours, then $\tau_{c-unif}^v(\tilde{G})$ has the same distribution as $\tau_{2-c-unif}^v(G) $, likewise $\tau_{par}^v(\tilde{G})$ and $\tau_{L-par}^v(G)$ are also equidistributed. Theorem \ref{conttimepar} is then applied to $\tilde{G}$ yielding $\tau_{2-c-unif}^v(G) = (1+o(1))\tau_{L-par}^v(G)$ w.h.p.\ and in expectation. Combining these relations yields \[ \tau_{L-par}^v(G) = \tau_{2-c-unif}^v(G)= (2+o(1))\cdot \tau_{c-unif}^v(G)= (2+o(1))\cdot \tau_{par}^v(G),\] w.h.p.\ and in expectation. For a general graph $G$ and $v \in V$ the exact same argument works however each of the equalities above holds only up to a $\Theta(1)$ factor, the result follows. 
	\end{proof}

	\section{Fundamental Networks}\label{fundamental}
	In this section we determine the dispersion for many well known graph topologies. 
	
	\subsection{The Complete Graph}

	We shall begin with the clique as this is most simple to analyse.
	
	\begin{thm}\label{complete}
		Let $K_n$ be the complete graph on $n$ vertices and $\kappa_{cc}$ be as in Lemma \ref{bren}. Then 
		\begin{align*}
		t_{par}(K_n) &\sim \frac{\pi^2 }{6}\cdot n 
		\qquad \text{and}\qquad t_{seq}(K_n) \sim \kappa_{cc}\cdot n,
		\end{align*}where 
		\[\kappa_{cc}:= \sum\limits_{i=1}^{\infty} \left( \frac{2}{i(3i-1)}-\frac{2}{i(3i+1)} \right) \approx 1.255. \] 
	\end{thm}
	
	Before proving the above we state a result needed to treat the Sequential-IDLA on cliques.
	\begin{lem}[\cite{Brennan}]\label{bren}Let $T:=T_n$ be the maximum of $n$ independent geometric random variables with parameters $\frac{i}{n}$ for $1\leq i\leq n$. Then the limit $\lim_{n \rightarrow \infty} \Ex{T}/n$ exists and is equal to $\kappa_{cc}$ 
	\end{lem}
	
	The constant $\kappa_{cc}$ is related to the longest wait time in the coupon collector process \cite{Brennan}. 
	
	\begin{proof} [Proof of Theorem \ref{complete}]
		Instead of analyzing the parallel process, we analyze the continuous-time Uniform-IDLA process (CTU-IDLA), in which each particle has a exponential clock of rate 1, and moves every time the clock rings until the particle settles. By Theorem 4.8 we have that the dispersion time of the Parallel-IDLA process and the CTU-IDLA process are asymptotically equal as long as the dispersion time of the Parallel-IDLA is $\lohm{\log n}$ w.p. $1-\lo{1/\log^2 n}$. The property holds trivially because as the last particle in the Sequential-IDLA takes geometric time of mean $n$ to settle, this holds also for the Parallel-IDLA due to the stochastic domination $\tau_{seq}^v\preceq \tau_{par}^v$ by Theorem~\ref{PStocS}. The analysis of the CTU-IDLA is quite simple: since particles move in continuous-time no two particles settle at the same time. Suppose there are $k$ unsettled particles, then the time needed until one of the $k$ particles settles in one of the $k$ unoccupied vertices is exponentially distributed with mean $(n-1)/k^2$. Summing up from $k=1$ to $n-1$ we obtain that the expected dispersion time is asymptotically $n\sum_{k \geq 1}k^{-2} = n\cdot (\pi^2/6-o(1))$.
		
		For $t_{seq}$ the longest walk in the Sequential-IDLA on $K_n$ is the longest waiting time in the Coupon Collector problem. This time is distributed as the maximum of $n$ independent geometric random variables with parameters $\frac{n-i+1}{n}$ for $1\leq i\leq n$. The result follows from Lemma \ref{bren}. \end{proof}
	\begin{rem}
		Observe that $\kappa_{cc} \approx 1.255$ and $\pi^2/6 \approx 1.645$ so the two constants are distinct. \end{rem}

	\subsection{The Path}
	Let $P_n$ be the path with $n$-vertices. Interestingly, the path provides an example where the sequential and parallel dispersion process take the same time up to lower order terms.
	\begin{thm}\label{PathConstSame}
		Let $M$ be the maximum of $n$ independent random variables representing the hitting time of a random walk to the vertex $n$, starting from $1$ on $P_n$. Then for the dispersion time,
		\[
		t_{seq}(P_n) = (1 \pm o(1)) \cdot \Ex{M}=t_{par}(P_n).
		\]
	\end{thm}
	\begin{proof}
		In the following, we will denote by $t_{seq}(m)$ the expected running time of the Sequential-IDLA on a path with $m$ vertices, when the source is the endpoint labelled by $1$. In the following, let $Y_1,Y_2,\ldots,Y_n$ be a collection of $n$ independent random variables, each of which describing the hitting time of a random walk from endpoint $1$ to $n-n/\log n$ (thus $Y_i=\tau_{hit}(1,n-n/\log n)$). In particular, these random walks will not settle and are therefore completely independent.
		
		The proof will be based on the following chain of inequalities: \newline
		{\small
			\begin{eqnarray*}
				t_{seq}\left(n-\frac{n}{\log n}\right) \stackrel{(1)}{\leq} t_{par}\left(n-\frac{n}{\log n}\right)\stackrel{(2)}{\leq} \Ex{\max_{1 \leq i \leq n-\frac{n}{\log n}} Y_i }&\stackrel{(3)}{\leq}&
				(1+o(1)) \cdot\Ex{ \max_{1 \leq i \leq \frac{n}{\log n}} Y_i} \nonumber\\
				&\stackrel{(4)}{\leq}& (1+o(1)) \cdot t_{seq}(n),
			\end{eqnarray*}
		}
		\normalsize
		and then finally
		\begin{align*}
		t_{seq}(n) \stackrel{(5)}{\leq} (1+o(1)) \cdot t_{seq}\left(n-\frac{n}{\log n}\right),
		\end{align*}
		and if all these inequalities hold, the claims of the theorem are established.
		
		Note that inequality (1) is a direct consequence of Theorem~\ref{PStocS}, and inequalities (2) and (4) follow directly from the definition of the Parallel-IDLA and Sequential-IDLA, respectively. Thus it only remains to prove (3) and (5).
		
		We first prove (3) - in fact, for notational convenience we will establish the stronger claim
		\[
		\Ex{ \max_{1 \leq i \leq n} Y_i} \leq (1+o(1)) \cdot \Ex{\max_{1 \leq i \leq \frac{n}{\log n}} Y_i},
		\]
		i.e., on the left hand side, we take the maximum over $n$ random variables instead of just $n-n/\log n$.
	
		To simplify notation, define $\tilde{Y}:=\max_{1 \leq i \leq n/\log n} Y_i$ and 
		define $Y:=\max_{1 \leq i \leq n} Y_i$. 
		In order to prove that $\Ex{\tilde{Y}}$ and $\Ex{Y}$ are close, consider a coupling where we first expose the values of the set $\{Y_1,Y_2,\ldots,Y_n\}$ and then assign those values through a random permutation. Next define by $F$ the random variable counting the $Y_i$'s which are larger than $\tilde{Y}$, in symbols,
		\[
		F := \left| \left\{ n/\log n < j \leq n \colon Y_j > \tilde{Y} \right\} \right|.
		\]
		Next note that for any $\lambda \geq 1$,
		\[
		\Pro{F \geq \lambda \cdot \log n} \leq \prod_{i=1}^{n/\log n} \left(1 - \frac{ \lambda \log n}{n-i+1} \right) \leq \prod_{i=1}^{n/\log n} \exp\left(-  \frac{\lambda \log n}{n} \right) \leq \exp\left(-\lambda \right).
		\]The first inequality holds by considering the probability that random ordering does not ``choose'' any of the $\lambda \log n$ longest walks for one of the first $n/\log n $ walks. I.e. if we have chosen $k$ so far, non of them being one of the $\lambda \log n $ longest, then we choose a long walk next time w.p. $(\lambda \log n )/k $. Thus for $\lambda = 2 \log n$, $\Pro{F \geq 2 \log^2 n } = n^{-2}$.
		
		Consider now the gap between the $(2 \log^2 n)$-th largest element of the values $\{Y_1,Y_2,\ldots,Y_n\}$ and the maximum. To this end, we will use the principle of deferred decisions and expose the $n$ trajectories in parallel order and stop as soon as there at most $2 \log^2 n$ walks which have not hit the other endpoint.
		
		Hence suppose we order these values such that w.l.o.g. $Y_1 \leq Y_2 \leq \cdots \leq Y_n$. Then for any $j \geq n- 2 \log^2 n$, the random variable $ Y_j - Y_{n-2\log^2 n}$ is stochastically dominated by one plus the hitting time from $1$ to $n$, so in particular, $\Ex{Y_j - Y_{n-2\log^2 n}}=\mathcal{O}(1+n^2)$. Furthermore, using the fact that from any start point, a random walk reaches the vertex $n$ is at most $2n^2$ steps with probability at least $1/2$, it follows that for any $\lambda > 0$, 
		\begin{equation}\label{probpath}
		\Pro{ Y_{j} - Y_{n-2\log^2 n} > 1+\lambda \cdot 2 n^2  } = \mathcal{O}(2^{-\lambda}).
		\end{equation}
		Choosing $\lambda = C \log \log n$ for some large constant $C > 0$, it follows by the Union bound over the at most $\log^2 n$ indices $j \in F$ that 
		\[
		\Pro{ Y \geq \tilde{Y} + \mathcal{O}(n^2 \log \log n) \, \mid \, F \leq 2\log^2 n} \leq \frac{1}{\log^2 n}.
		\]
		To conclude, it follows by the Union bound that w.p. at least $1-3/(\log n)^{2}$, our coupling satisfies
		\[
		Y - \tilde{Y}  \leq C \cdot n^2 \log \log n.
		\]
		Otherwise, we still have $\Ex{Y - \tilde{Y}  \, \mid \, \mathcal{E} } = \mathcal{O}(n^2 \log n) + \mathcal{O}(n^2 \log \log n)$, where $\mathcal{E}$ denotes the event that any of the above probabilistic arguments fail. The result follows since $\Pro{\mathcal{E}} = \BO{1/\log^2 n}$.

		We now continue to prove inequality (5).
		To this end we will construct a coupling between the $n$ walks in $t_{seq}(n)$ and the $n-n/\log n$ walks in $t_{seq}(n-n/\log n)$.
		Consider the first $n/\log n$ random walks in the $t_{seq}(n)$ setting. For each of them, the expected time to settle is $\mathcal{O}(n^2/\log^2 n)$ and by an argument similar to \eqref{probpath}, none of them will take more than $\mathcal{O}(n^2)$ with probability $1-n^{-\omega(1)}$.
		
		The trajectories of the next $n-n/\log n$ walks of $t_{seq}(n)$ can be coupled with the ones in $t_{seq}(n-n/\log n)$, so if a walk moves from vertex $x$ to $x+1$ in $t_{seq}(n-n/\log n)$, then the corresponding walk in $t_{seq}(n)$ moves from $x+n/\log n$ to $x+1+n/\log n$. The only difficulty arises when the walk in $t_{seq}(n)$ is at a vertex between $1$ and $n-n/\log n$. To capture this, we will consider so-called excursions which are epochs in which the random walk is at such a vertex. Notice that the total number of steps that are taken as part of any excursion is at most the total number of visits to any vertex in $1,2,\ldots,n/\log n$. However, note that the expected number of visits to any of these vertices is $\mathcal{O}(n \log n)$ for a random walk of $\mathcal{O}(n^2 \log n)$ steps, and thus by a standard Chernoff Bound for random walks, it follows that any of these vertices is visited at most $\mathcal{O}(n \log n)$ times with probability at least $1-n^{-2}$. Thus by the Union bound, the total number steps spend in any excursion is at most $\mathcal{O}(n^2)$ with probability at least $1-n^{-1}$.

		To conclude, we have shown that with probability at least $1-n^{-1}$ there is a coupling between $\tau_{seq}(n)$ and $\tau_{seq}(n-n/\log n)$ such that
		\[
		\tau_{seq}(n) \leq \tau_{seq}(n-n/\log n) + \mathcal{O}(n^2).
		\]
		Note that we can verify whether this coupling holds by inspecting only the first $\mathcal{O}(n^2 \log n)$ steps of the random walks. Thus even conditional on the coupling failing, we have $t_{seq}(n) = \mathcal{O}(n^2 \log n)$. Since $t_{seq}(n) = \Omega(n^2)$, it follows that for the expected values,
		\begin{align*}
		t_{seq}(n)  &\leq (1+o(1)) \cdot t_{seq}(n-n/\log n).
		\end{align*}
		
	\end{proof}

	\subsection{Expanders and The Hypercube}
	We call a graph an expander if $1-\lambda_2 = \Omega(1)$, where $\lambda_2$ is the second largest absolute eigenvalue. 
	
	\begin{thm}\label{expandthom}Let $G$ be an $n$-vertex almost-regular expander graph. Then $t_{seq}(G),t_{par}(G) =\Theta(n)$.
	\end{thm} 
	
	\begin{proof}The lower bound for $t_{seq}$ follows from Theorem \ref{lowerbound}. The upper bound on $t_{par}(G)$ follows from Corollary \ref{parlam2}. The result then follows since $t_{seq}(G)\leq t_{par}(G)$ by Theorem \ref{PStocS}.  
	\end{proof}
	\begin{rem}
		In particular this result covers (w.h.p.) random $d$-regular graphs, for fixed $d$, and the binomial random graph $\mathcal{G}(n,p)$ above the connectivity threshold, when $np \geq c\log(n)$, $c>1$.  
	\end{rem}
	The Hypercube $H_{d}$, where $n=2^d$, is the graph where each vertex is a binary string of length $d$ and two vertices are connected if their associated binary strings differ in one digit. The hypercube is not an expander since $1-\lambda_2 = 1/d= 1/\log_2 n $ however, we still achieve a linear bound.
	
	\begin{thm}
		Let $H_d$ be the hypercube with $n=2^d$ vertices. Then $t_{seq}(H_n),t_{par}(H_d) =\Theta(n)$. 
	\end{thm}\begin{proof}
		The lower bound for $t_{seq}$ follows from Theorem \ref{lowerbound}. Due to Theorem~\ref{PStocS} we only need to find an upper bound for $t_{par}$. As laziness only changes the dispersion time by a constant factor, we work with lazy walks. For the upper bound we seek to apply Theorem \ref{partialparupper} however, unlike in Theorem \ref{expandthom}, we shall use an argument based on return probabilities in $H_n$ to bound hitting times rather than appealing to Lemma \ref{setestimate}. Also note that, since the sum in Theorem~\ref{PStocS} only has $\mathcal{O}\!\left(\log n \right)$ terms and by monotonicity of hitting times of sets, it will be sufficient to cover the case $1 \leq |S| \leq (\log n)/2$. If we can prove that the hitting time is $\mathcal{O}(n/|S|)$ in this case then we are done. We divide time into epochs of length $2\log^2 n$ and prove the probability we hit $S$ in one epoch is at least $\Omega((\log n)^2|S|/n)$. In the first $\log^2 n$ steps of an epoch we allow the walk to mix ignoring if the walk hits $S$. Then, with high probability we can couple our walk with the stationary distribution. In the second $\log^2 n$ steps of an epoch we observe if the walk hits $S$. Let $Z$ the random variable which counts the number of visits to the set $S$ in $(\log n)^2$ steps. Then $\mathbf{Pr}_{\pi}[\tau_S \leq (\log n)^2] = \mathbf{Pr}_{\pi}[Z \geq 1]$ and
		$$\mathbf{Pr}_{\pi}[Z \geq 1]= \frac{\mathbf{E}_{\pi}[Z]}{\mathbf{E}_{\pi}[Z|Z\geq 1]} \geq \frac{(\log n)^2|S|/n}{\max_{u \in S}\sum_{t=0}^{(\log n)^2}\tilde p_{u,S}^t}.$$
		\begin{clm}\label{clmreturns}
			For any set $S$ and $u\in S$ if $|S|\leq (\log n )/2$, then  $\sum_{t=0}^{(\log n)^2}\tilde p_{u,S}^t=\BO{1}$.   
		\end{clm}
		Thus by Claim \ref{clmreturns} (proved later) $\mathbf{Pr}_{\pi}[\tau_S \leq (\log n)^2]= \BOhm{\Omega((\log n)^2|S|/n) }$ and so $t_{hit}(\pi ,S)\leq \BO{\max \{n/|S|, n/\log n  \}}$. Thus as discussed earlier the upper bound follows from Theorem \ref{partialparupper}. 
	\end{proof}
	
	The proof of Claim \ref{clmreturns} will make heavy use of \cite[Lem.\ 7]{CooperFriezeVacant}, we paraphrase it hear for convenience: 
	\begin{lem}\label{hypercubewalk}
		Let $\mathcal{W}(i)$, $i\geq 0$ be the lazy walk in $H_d$ and $T=(\log n)^2$. Then, for any $v \in V$,
		\begin{enumerate}[(i)]
			\item
			$\displaystyle{R_v :=\sum_{i=0}^T \tilde{p}_{u,u}^t= 2+\frac{2}{d}+\BO{\frac{1}{d^{2}}}.}$
			\item Suppose $W(0)$  is at distance at least 2 from $v$ (resp. at least 3 from $v$). The probability $W$ visits $\Gamma(v)$ within $L=O(T \log n)$ steps is $P(2,L)=O(1/d)$ (resp. $P(3,L)=O(1/d^{2})$).
			\item Let $C\subseteq N(v)$. For a walk starting from $u \in C$, let $R_C$ denote the expected number of returns to $C$ within $T$ steps. Then, in the lazy walk,   $R_C=2+\BO{1/d}$.
		\end{enumerate}
	\end{lem}
	
	\begin{proof}[Proof of Claim \ref{clmreturns}] Let $C=\{u \}\cup (\Gamma(u)\cap S) $. Let $R_2(u,t)$ (resp. $R_{\geq 3}(u,t)$) be the expected number of visits to a vertex at distance $2$ (resp. distance $\geq 3$) from $u$ before time $t$. Then if $T=(\log n )^2$ we have 
		\begin{equation}\label{returnsbdd33}\sum\limits_{i=1}^{(\log n)^2} \tilde{p}_{u,S}^i \leq  R_C +\frac{\log n}{2}\cdot R_2(u,T) + \frac{\log n}{2}\cdot  R_{\geq 3}(u,T), \end{equation} where the first term counts returns to $u$ and the portion of $S$ in $u$'s neighbourhood, the second term counts visits to members of $S$ at distance $2$ and the third to those at distance $3$ or greater (where we recall that $|S|\leq (\log n)/2$ .
		
		We have the crude bound $R_2(u,T)\leq P(2,T)\cdot R_v $, where $v$ is any vertex by transitivity of $H_d$. Thus by Lemma \ref{hypercubewalk},  $R_2(u,T)\leq \BO{1/d}\cdot (2+\BO{1/d}) = \BO{1/d}$. Similarly $R_{\geq 3}(u,T)\leq P(3,T)\cdot R_v = \BO{1/d^2}$. Note $ R_{C} = \BO{1}$ by Lemma \ref{hypercubewalk} (ii). It follows from \eqref{returnsbdd33} that $\sum_{i=1}^T \tilde{p}_{u,S}^i= \BO{1} $.  
	\end{proof} 
	\subsection{Tori and Grids} 
	Let $B(r):= \left\{\mathbf{x} \in \mathbb{Z}^d : x_1^2 +\dots + x_d^2 \leq r^2 \right\}$ be the ball of radius $r$ in $\mathbb{Z}^d$.	
	\begin{lem}\label{slowwalklem} Let $d=1,2$ be fixed. For any $\beta >0$ there exists some $C>0$ such that the random walk of length $Ct\log t$ from the origin in $\mathbb{Z}^d$ does not exit $B(\sqrt{t})$ with probability at least $1/t^\beta$. 
	\end{lem}	
	\begin{proof}
		
		Let $S_j$ be the position of a random walk at time $j$ started from $0$. For $t>0$ let $\mathcal{E}_0$ be the event $\left\{S_j \in B(\sqrt{t}/2) \text{ for all } 0\leq j \leq c^2t-1  \right\}$. By the Central Limit Theorem \cite{LawInter} for all $\varepsilon >0$ there is some $c >0$ such that for large $t$
		\[\Pro{\frac{S_{ c^2 t}}{\sqrt{c^2 t}} \not\in B\left(\frac{  \sqrt{t}/2 }{\sqrt{c^2 t}} \right)} \leq \left(1+\mathcal{O}\!\left(\frac{1}{\sqrt{t}}\right)\right) \int_{\mathbb{R}^2\backslash B\left(\frac{1}{ 2c}\right)}\frac{e^{-|\mathbf{x}^2|}}{\pi} \;\mathrm{d}\,\mathbf{x} \leq \varepsilon.    \] Thus by the Reflection Principal \cite[Prop. 1.6.2]{lawlermodern} the probability a random walk stays within the ball $B\left(\sqrt{t}/2\right)$ for $c^2t$ units of time is at least $1-2\varepsilon$. For $i\geq 1$ let $\mathcal{E}_i$ be the event 
		\[\left\{ S_j \in B(\sqrt{t}) \text{ for all } i\cdot c^2t\leq j \leq (i+1)\cdot c^2t-2  \right\}\cap \left\{S_{(i+1)\cdot c^2t-1 }\in B(\sqrt{t}/2) \right\}.\]By geometric considerations we see that $\Pro{\mathcal{E}_{i+1}|\mathcal{E}_{i}} \geq \left(1-2\varepsilon\right) /2d \geq 1/(2d+1)$ for small enough $c>0$. Observe that $\left\{S_k \in B(\sqrt{t}) \text{ for all } 0\leq  k\leq \alpha c^2t \log t  \right\}\supseteq \bigcap_{i=0}^{\alpha \log t} \mathcal{E}_i$, for any $\alpha >0$.   Thus for any fixed $\beta >0$ provided $\alpha\leq \beta/\log(2d+1)$ we have 
		\[\Pro{S_k \in B(\sqrt{t}) \text{ for all } 0\leq  k\leq \alpha c^2t \log t}\geq (1-2\varepsilon) \left( \frac{1}{2d+1}\right)^{\alpha \log t } \geq \frac{1}{t^\beta  }.   \] The result follows by taking $c>0$ small enough.
	\end{proof}
	\begin{thm}\label{cycle} 
		For the path/cycle, $\Theta(n^2 \log n)$ steps are needed in expectation and with probability at least $1-o(1)$.
	\end{thm}
	\begin{proof}
		The upper bound for either graph follows from Lemma \ref{lem:general}. For the lower bound in the cycle if at some time an interval $[-a,b]$ has been settled around the origin then by the gamblers ruin formula the end point closest to the origin receives the next particle with probability at least $1/2$. Thus by a Chernoff bound w.h.p.\ after $2n/3$ particles have settled the interval $[-n/4,n/4]$ is occupied. Each of the remaining $n/3$ particles must exit the ball $B(n/4)$ in order to settle. Thus by Lemma \ref{slowwalklem} there is some $C>0$ such that the probability that one walk takes longer than $Cn^2\log(n)$ to exit $B(n/4)$ is at least $1- (1-1/n^\beta)^{n/3} = 1-o(1)$. The result for the path by similarly considering a return to the origin as a change in parity for a walk on the cycle and adding settled vertices to both ends simultaneously. 
	\end{proof}
	The next result does not settle the dispersion time on the two-dimensional grid, but improves on the trivial $\Omega(n)$ bound.
	\begin{pro}
		Let $G$ be either the finite box $\left[-\lfloor\sqrt{n}/2\rfloor , \lfloor\sqrt{n}/2\rfloor \right]^2\subset \mathbb{Z}^2$ in the two-dimensional grid, or the two-dimensional finite torus on $n$ vertices. Then $t_{seq}(G),t_{par}(G)=\Omega(n \log n)$.
	\end{pro}
	\begin{proof}
		We will prove the lower bound for $t_{seq}$ only, since the corresponding lower bound for $t_{par}$ will follow from $t_{par} \geq t_{seq}$.
		
		Let $A(t)$ denote the aggregate of the Sequential-IDLA once $t$ particles have settled. Theorem 1 of \cite{Sheff1} states that for each $\gamma$ there exists an $a = a(\gamma
		) < 1$ such that for all sufficiently large $r$, \begin{equation}\label{shapethmeq}\mathbb{P}\left[B( r-a\log r)		\subseteq \mathcal{A}(\pi r^2)		\subseteq  B(r+a\log r)\right]\geq 1- r^{-\gamma}.\end{equation} We can couple the process on $G$ with the process on $\mathbb{Z}^2$ up until the point $t^*$ when the first particle settles a vertex on the boundary (or wraps around in the torus). By \eqref{shapethmeq} we can condition on the aggregate $A(t^*)$ containing a ball of radius $\lfloor\sqrt{n}/2\rfloor - a\log n $ w.h.p., for some $a<\infty $. Thus the remaining $n-t^* > (1-\pi/4)n>n/5$ particles must all exit the ball $B(\sqrt{n}/3)$ before settling. Now by Lemma \ref{slowwalklem} the probability that one walk takes longer than $Cn\log n $ to do this is at least $1- (1-1/n^\beta)^{n/5} = 1-o(1)$. The result follows.  
	\end{proof}	
	
	\begin{thm}
		Let $G$ be the $d$-dimensional torus/grid where $d\geq3$. Then $t_{seq}(G),t_{par}(G) =\Theta(n)$.
	\end{thm}
	\begin{proof}
		The lower bound for $t_{seq}$ follows from Theorem \ref{lowerbound}. For the $d$-dimensional torus/grid we have the well-known bound $p_{u,v}^t \leq 1/n + \mathcal{O}(t^{-d/2})$. This estimate applied in combination with Lemma \ref{setestimate} to Theorem \ref{partialparupper} implies a bound of $\mathcal{O}(n)$ on the dispersion time whenever $d \geq 3$. 
	\end{proof}
	
	\subsection{Binary Tree}
	In this section we consider the binary tree $T_n$ with $n+1 = 2^k-1$ vertices, $2^{k-1}$ leaves and root $r$. 
	
	\begin{thm}\label{thm:binarytree}
		For the binary tree $T_n$ with root $r$, we have $\tau_{seq}^r,\tau_{par}^r  =\BT{n(\log n )^2}$ w.h.p. and in expectation, consequently $t_{seq}, t_{par}=\BT{n(\log n)^2} $. 
	\end{thm}
	Recall that the hitting time in the Binary tree with $n$ vertices is $\mathcal{O}(n\log n)$. Thus, by Theorem \ref{theorem:general}, the dispersion time of the Parallel-IDLA process is $\mathcal{O}(n(\log n)^2)$ w.h.p.\ and in expectation. Thus to establish Theorem \ref{thm:binarytree} it remains to show that the dispersion time of Sequential-IDLA is at least $\Omega( n(\log n)^2)$ w.h.p., proving that $t_{seq} = \Theta(t_{par}) = \Theta(n(\log n)^2)$, due to Theorem~\ref{PStocS}. 
	
	To prove the lower bound we show the last $poly(n)$ unoccupied vertices are clustered in such a way that one of the last $poly(n)$ walks will have trouble finding the cluster. The first step of this strategy is to establish the following lemma, which in some sense a shape theorem for the binary tree.  
	
	\begin{lem}\label{treeLRlemma}
		Consider a complete binary tree with $n=2^k-1$ vertices and the root $r$ being the source of the Sequential-IDLA. Let $\tau$ be the first time when one of the two sub-trees with $2^{k-1}-1$ vertices is completely filled and fix $0< \varepsilon <1/4$. Then with probability at least $1-2n^{-\varepsilon}$, the other sub-tree still has at least $n^{\varepsilon}/(3\log_2 n)$ unoccupied vertices at time $\tau$.
	\end{lem}
	The lemma above allows us to show that after some time in the process all the remaining unsettled vertices are contained in a sub-tree of significant distance from the root. The next lemma says that w.h.p.\ one of the remaining walks takes a long time to enter the sub-tree.

	\begin{lem}\label{lem:tailgeom}
		Let $u$ be an arbitrary but fixed vertex which has distance $\epsilon \log_2 n$ from the root, where $0<\epsilon \leq 1 $ is some constant. For any given $c>0$ there exists $c'>0$ such that a random walk of length $c'\varepsilon n \log^2 n$ starting from the root $r$ visits $u$ with probability at most $1-n^{-c}$. 
	\end{lem}

	These two lemmas are the main technical component of this chapter and are proved in Sections \ref{sec:treeLRlemmaproof} and \ref{sec:tailgeomproof} respectively, first we shall prove Theorem \ref{thm:binarytree}. 
	
	\begin{proof}[Proof of Theorem \ref{thm:binarytree}]As mentioned above it suffices to prove a w.h.p. lower bound on $\tau_{seq}^r$. 
		Suppose $n = 2^k-1$ and let $r$ denote the root of the binary tree. Let $T_1,T_2,\ldots,T_{x}$ with $x=2^{\lfloor k/16 \rfloor} \leq n^{1/32}$ be a labelling of all sub-trees whose root is at distance $\lfloor k/32 \rfloor$ form the root $r$. Note each of those sub-trees has $2^{k-\lfloor k/32 \rfloor}-1 = (1+o(1))n^{31/32}$ vertices.
		
		Observe that whenever a particle enters to one of those sub-trees, we can imagine the filling process as an independent IDLA process on that sub-tree. Indeed, when a particle moves inside such a tree, it is moving as a random walk on the tree until it settles, and if the particle leaves the sub-tree from the root, we can imagine we pause the process until a new particle arrives again, restarting the process. From the previous observation we can  apply Lemma \ref{treeLRlemma} above with $\varepsilon = 1/8$, it follows that whenever we fill one of the sub-trees of $T_i$, the other sub-tree has at least $ \BOhm{\frac{n^{31\varepsilon/32}}{\log_2 n}}\geq n^{1/9}$ with probability at least $1-2n^{31\varepsilon/32}$. By the union bound, the above happens for all the sub-trees $T_1,\ldots, T_x$ at the same time with probability at least $1-2n^{31\varepsilon/32}n^{1/32} \geq 1-n^{-\Theta(1)}$.
		
		Now, consider all the sub-trees at distance $\lfloor k/32 \rfloor+1$ from the root, and suppose that just after settling the $i$-th particle, all but one of those sub-trees are filled.
		Without lost of generality, we assume that such a sub-tree is a sub-tree of $T_1$, and denote the left and right sub-trees of $T_1$ by $T_{11}$ and $T_{12}$. Additionally, suppose that that $T_{11}$ is filled just after settling the $i$-th particle. We conclude that after settling the $i$-particle, all trees $T_2,\ldots, T_x$ are filled, and since $T_{11}$ just became filled, we deduce that all the remaining unsettled vertices are located in $T_{12}$ and there are at least $n^{1/9}$ of them.  Choosing $c=1/10$ in Lemma~\ref{lem:tailgeom} below gives that one of the $M\geq n^{1/9}$ remaining walks take longer than $c'\varepsilon n \log^2 n$ hit the sub-tree $T_1$ with probability at least $1- \left(1-n^{-1/10}\right)^{n^{1/9}} \geq 1- o(1)$. Concluding that it takes $\Omega(n \log^2 n)$ to settle all particles w.h.p.	
	\end{proof}

	\subsubsection{Proof of Lemma~\ref{treeLRlemma}}\label{sec:treeLRlemmaproof}
	We begin with the following simple lemma needed to prove Lemma \ref{treeLRlemma}. 
	\begin{lem}\label{resistlemma}
		In the binary tree $T_n$ of height $k$, the probability that a fixed leaf $u$ is visited before the walk returns to the root $r$ is $1/ (2(k-1))$.
	\end{lem}
	\begin{proof} The formula $\Pro{\text{A Random Walk from $r$ hits $u$ before returning to }r}= \left(R(r,u)\cdot d(r) \right)^{-1}$ can be found in \cite[Prop. 9.5.]{levin2009markov}. The result follows since the resistance $R(r,u)$ in a tree is given by graph distance and the degree of the root, $d(r)$, is $2$.
	\end{proof}
	
	\begin{proof}[Proof of Lemma \ref{treeLRlemma}]

		We divide the Binary tree into a root, and a left and right sub-tree. To study the IDLA process, we consider the following algorithm. Consider an infinite sequence of (independent) random walks starting in the root of the left-tree. These walks finish when they hit the root of the original tree. We also consider an (independent) infinite sequence for the right sub-tree. To run the IDLA process, we start in the root of the binary tree and settle the first particle. From the second particle on, each time a particle is in the root it moves to the left or right sub-tree with probability $1/2$. The $i$-th time a particle moves to the left (right) sub-tree, it follows deterministically the $i$-th predetermined walk until it reaches a vertex for first time or returns to the root of the binary tree. The advantage of this procedure is that once we predetermine the infinite random walk sequences in the left and right sub-tree, we know the number of times particles need to move from the root either to the left sub-tree or to the right sub-tree in order to fill the left and right sub-trees respectively. Let us call such quantities, the number of visits to each sub-tree required to fill it, $L$ and $R$ (for the left and right sub-trees). Note that $n/2\leq L,R$ because we need to move at least $n/2$ times to the left (right) sub-tree in order to fill it. We prove the following property of the predetermined walks: Let $S$ be the number of walks needed to cover the last $n^{\varepsilon}/(3\log_2 n)$ unoccupied vertices of the left ( or right) sub-tree. Define the event $\mathcal E_1 = \{S<n^{\varepsilon}\}$. We prove that $\mathcal E_1$ occurs w.h.p., indeed,
		\begin{equation}\label{equat19}
		\Pro{S \geq n^{\varepsilon}} \leq \Pro{Bin\left(n^\varepsilon,\frac{1}{2\log_2 n}\right) \leq \frac{n^{\varepsilon}}{3\log_2 n}}\leq \exp\left(-\frac{n^\varepsilon}{72\log_2 n}\right).\end{equation}
		In the first inequality follows from Lemma \ref{resistlemma} (hitting a leaf is harder than hitting a non-leaf in a excursion), for second inequality we use Chernoff's bounds. Therefore, with probability at least $1-2\exp(-\frac{n^\varepsilon}{72\log_2 n})$, the last $n^{\varepsilon}$ walks cover at least $n^{\varepsilon}/(3\log_2 n)$ unoccupied vertices of the left (or right) sub-trees. 
		
		From now, we assume all the walk in the left (right) sub-trees are predetermined. Let $W_j$ be 1 if the $j$-th time a particle leaves the root moves to the left sub-tree, $W_j = 0$ otherwise. Denote $L_i = \sum_{j=1}^i W_j$ and $R_i = n-L_i$. Define $\tau =\min\{i\geq 1: L_i =L \text{ or } R_i =R\}$. 
		
		\begin{clm}\label{claim:towers} For $\varepsilon<1/4$ it holds that $\max\{R-R_{\tau}, L-L_{\tau}\} \geq n^{\varepsilon}$ with probability at least $1-n^{-\varepsilon}$.
		\end{clm}
		The proof of the claim is temporally deferred. The claim above essentially tells us that when we fill one sub-tree, the other needs at least $n^{\varepsilon}$ more walks to be filled with high probability. Denote $\mathcal{E}_2 = \{\max\{R-R_{\tau}, L-L_{\tau}\} \geq n^{\varepsilon}\}$. Note that the statement of this Lemma follows from proving that $\mathcal E_1 \cap \mathcal{E}_2 $ holds with probability at least $1-2n^{-\varepsilon}$. By \eqref{equat19} and Claim \ref{claim:towers} we have   
		\begin{equation*}
		\Pro{(\mathcal E_1 \cap \mathcal E_2)^c} \leq \Pro{\mathcal{E}_1^c}+ \Pro{\mathcal E_2^c} \leq 2\exp\left(-\frac{n^\varepsilon}{72\log_2 n}\right) +n^{-\varepsilon}+\leq 2n^{-\varepsilon}.
		\end{equation*}\end{proof}
	
	\begin{proof}[\textbf{Proof Of Claim~\ref{claim:towers}}] 
		Recall that after the $i$-th time a particle moves from the root to one of the sub-trees, we have $R_i+L_i = i$. Also, if such particle moves to the left sub-tree $L_i=L_{i-1}+1$ and $R_i = R_{i-1}$ (similarly if the particle moves to the right sub-tree). We can see the process as balls into 2 bins (left and right bins). At each round we allocate a ball to one of the bins at random. The process finishes when the left bin has $L$ balls or when the right bin has $R$ balls, but for convenience we allow the process to keep adding balls after such a point. We work with a continuous time version of this process where balls arrive to each bin following independent Poisson processes $N_l(t)$ and $N_r(t)$ of rate 1 for the left and right bin, respectively. Let $\tau_l$ (resp.\  $\tau_r$) be the first time $t$ such that $N_l(t) \geq L$ (resp.\ $N_r(t) \geq R$).
		Consider the time $\tau_l=t$ and consider the load of the other bin $N_r(t)$. First, note that as $L,R \geq n/2$, the event $\{t \leq 15 n^{4 \epsilon}\}$ occurs only with probability at most $\exp(-n^{\Omega(1)})$ (using a Chernoff bound) and therefore in the remainder of the proof we will assume $t \geq 15n^{4 \epsilon}$. Note that for any $\tau_l = t$, the load of the right bin is exactly a Poisson random variable with parameter $t$. For any integer $x$, $\Pro{Poi(t) = x} \leq 2/\sqrt{2\pi t }$ thus using the lower bound on $t$
		\[
		\Pro{  |R-N_r(\tau_l)| < n^{\epsilon}    |\tau_l=t} \leq 2n^{\epsilon} \cdot (2/\sqrt{2\pi t }) \leq n^{-\epsilon}/2.
		\]Analogous arguments work for  $\tau_r$ and $|L-N_l(\tau_r)|$. By the union bound the result holds. \end{proof}

	\subsubsection{Proof of Lemma~\ref{lem:tailgeom}}\label{sec:tailgeomproof}
	\begin{lem}\label{lemma:hittingunlog2n}
		Let $c>0$ be fixed. Then, a random walk $(X_t)$ of length $n\lceil c (k-1)^2 \rceil/3$ on $T_n$ starting from the root $r$ visits an arbitrary but fixed leaf $u$ w.p. at most $1-e^{-c/2}\cdot n^{-c/(2\log 2)}$.
	\end{lem}
	\begin{proof}
		First note that by Lemma \ref{resistlemma}, it follows that a random walk does not visit leaf $u$ before the $\lceil c (k-1)^2 \rceil$-th return to the root with probability at least
		\[
		\left(1 - \frac{1}{2(k-1)} \right)^{\lceil c (k-1)^2 \rceil} \geq e^{-c k/2 - (c/4)(1+o(1)) } \geq \left(\frac{1}{n}\right)^{c/(2\log 2)}\cdot e^{-c/3},
		\]
		where the second inequality due to the fact that $n-1 = 2^k-1$ and thus $\log n = k \log 2$.
		
		Consider now a random walk of length $\ell = dn \lceil c (k-1)^2 \rceil $ for some constant $d>0$. We wish to show we have that not too many excursions (visits to the root $r$) during $\ell$ times w.h.p.\ Let $L$ be the set of leaves and $r$ be the root. Let $\tau_{A}$, ($\tau_{A}^+$) be the first hitting (return) time of the vertex/set $A$ by the random walk $X_t$. By \cite[Prop. 9.5.]{levin2009markov} we have 
		\begin{equation}\label{resistbounds}
		\mathbf{P}\left[ \tau_{L}< \tau_{r}^+ \big| X_0 = r\right]= 1/(2\cdot 2)  =1/4\quad \text{and} \quad \mathbf{P}\left[ \tau_{r}< \tau_{L}^+ \big| X_0 \in L \right]= 1/(2\cdot (n/2)) = 1/n.
		\end{equation}To simplify the analysis we shall consider only times when the walk is at the root or the leaves reducing the tree to a two state Markov chain. Indeed, we start the walk at the root and say it jumps to a leaf w.p. $1/4$, once at a leaf it can jump to the root w.p. $1/n$. 
		
		To bound the number of visits to $r$ from above we can assume that each attempt to get from $r$ to $L$ (or $L$ to $r$) takes at least $2$ units of time. Thus we have at most $\ell/2 $ tries to hit the root from $L$ and the number of successes is dominated by a Binomial r.v. with parameters $\ell/2$ and $1/n$ by \eqref{resistbounds}. Thus we hit $r$ from $L$ at most  $ t_1=(1+ 1/10) \ell/(2n) $ times w.p. $1-n^{-\omega(1)}$ by a Chernoff bound. Let $R_i$ be the number of returns to $r$ by the random walk from $r$ on its $i^{th}$ trip to $r$ before hitting $L$ again, note that $R_i$ is geometrically distributed with parameter $1/4$ by \eqref{resistbounds}. Thus if we let $Y(t_1)= \sum_{i=0}^{t_1}R_i$ be the number of returns to $r$ during a random walk of length $\ell$ conditional on $t_1$ successful returns to $r$ from $L$, then by Lemma \ref{Erlconc} \eqref{itmsumgeo}
		\[\Pro{Y(t_i) > 9t_1/2 } \leq \exp\left(- 9t_1\left(1- 8/5\right)^2/4  \right),
		\]
		holds for any fixed $d>0$. Then, by taking $d = 1/3$, with probability at least $1- 2n^{-\omega(1)}$ the number of returns to $r$ (and excursions from $r$) is bounded by $9t_1/2 \leq (9/2)\cdot (11/10) \ell /(2n) < \lceil c (k-1)^2 \rceil $.
		
		Thus we have 
		\begin{align*}
		\Pro{\tau_{hit}(r,u)> n \lceil c (k-1)^2 \rceil /3} &\geq \Pro{\text{$X_t$ does not visit $u$ in the first $\lceil c (k-1)^2 \rceil$ excursions}}\\
		&\qquad - \Pro{\text{There are more than $\lceil c (k-1)^2 \rceil$ excursions}} \\
		&\geq  e^{-c/3}\cdot n^{-c/(2\log 2)} - 2n^{-\omega (1)}.
		\end{align*}  The proof follows from noting the above is greater than  $e^{-c/2}\cdot n^{-c/(2\log 2)}$ for large $n$.\end{proof}
	
	Finally, we can now extend the result from the previous lemma to internal vertices, and prove the Lemma \ref{lem:tailgeom}, the a key Lemma about hitting time of clustered sets.

	\begin{proof}[Proof of Lemma~\ref{lem:tailgeom}]
		Let $\tilde{T}$ be the top of the binary tree $T$, this is the tree induced by all vertices that have distance at most $\epsilon \log_2 n$ from the root. Let $\tilde{L}$ be the set of leaves in $\tilde{T}$. By Lemma~\ref{lemma:hittingunlog2n}, we know from that given $c>0$, a random walk of length $c\varepsilon^2 n^{\epsilon} \log^2 n/3$ on $\tilde{T}$ does not visit a vertex $u\in \tilde{L}$ with probability at least $n^{-c\epsilon/(2\log 2)}$.  By the random walk Chernoff bound \cite{MitzChern} the random walk on $\tilde{T}$ makes at least $\nu = c\varepsilon^2 n^{\epsilon} \log^2 n/7$ visits to $\tilde{L}\backslash \{u\}$ with probability at least \[1 -\sqrt{n^{\varepsilon} }\cdot \exp\left( - \frac{(1/7)^2 \cdot c\varepsilon^2 n^{\epsilon} \log^2 n}{6\cdot 72 \cdot  t_{mix}\left(\tilde{T} \right)}  \right) = 1- n^{\omega(1)}.\]We will couple the walk on $\tilde{T}$ to a longer walk on the tree $T$ by allowing the walk to continue into sub-trees pendant to $\tilde{L}$. Let $S= \sum_{i=1}^{\nu }V_i$ be the amount of time spent in the sub-trees pendent to $\tilde{L}$ by the coupled walk, where $V_i$ is the amount of time spent in a pendent sub-tree before returning to $\tilde{L}$ for the $i^{th}$ time. Now by \eqref{resistbounds} a random walk in $T$ from $l\in \tilde{L}$ goes into the sub-tree pendant from $l$ and does not return to $l$ for at least $n^{1-\varepsilon}$ steps with probability $(2/3)\cdot (1/4)\cdot (1-1/n^{1-\varepsilon})^{n^{1-\varepsilon}} \sim  1/(6e)$. Since the amount of time spent by the walks in each sub-tree is identically distributed $S \geq \nu/(7e)\cdot n^{1-\varepsilon} = c\varepsilon^2 n \log^2 n/(7^2e)$ with probability $1- e^{-\Omega(n^\varepsilon)} $ by a Chernoff bound. Combining the above a walk of length $c\varepsilon^2 n \log^2 n/(7^2e)$ on $T$ hits $u$ with probability at most $ 1-e^{-c/2}\cdot n^{-c\epsilon/(2\log 2)}- n^{\omega(1)}- e^{-\Omega(n^\varepsilon)} \leq 1- n^{-c\epsilon/(2\log 3)}. $
	\end{proof}

	\subsection{The Lollipop}
	Let $L_n$ be the lollipop graph, which consists of a $\lceil n/2\rceil $ vertex clique $K$ attached at a vertex $v\in K$ by a single edge to the endpoint of a path $P$ of length $\lfloor n/2\rfloor $.
	\begin{pro}\label{lolli} Let $u\in K$, $u\neq v$. Then,   $\tau_{seq}^u(L_n) = \Omega\left( n^3\cdot \log(n)\right)$ w.h.p.\  . 
		
	\end{pro}
	\begin{proof}Let $w$ be a vertex half way down the path and $\mathcal{E}$ be the event that a walk from a vertex in $K\backslash \{v \} $ hits $w$ before returning to $K\backslash \{v \} $. For $\mathcal{E}$ to occur the walk must hit $v$, walk one step in the path then hit $w$ before returning to $K\backslash \{ v\} $, thus $\Pro{\mathcal{E}} \leq (2/n)\cdot (2/n)\cdot (4/n)\cdot ( 1-2/n)  \leq 9/n^3.$ During the sequential process $n/4$ vertices must hit $w$ before settling and that by the time $w$ is first hit the clique $K$ is fully occupied w.h.p.. Conditional on this we can lower bound $\tau_{seq}^u$ by the expected number of trials it takes for the longest of the last $n/4$ walks to hit $w$. For each walk such a trial is described by the event $\mathcal{E}$ and thus the number of trials required by one walk dominates a $\geo{9/n^3}$ random variable. Hence we have \[\Pro{\text{walk $i$ needs more than }n^3\log(n)/18 \text{ trials }  } \geq \left(1 - 9/n^3\right)^{n^3\log(n)/18} \geq 1/\sqrt{n}.  \] 
		Thus the probability all of the last $n/4$ walks need less than $n^3\log(n)/18$ trials is less than $\left(1- 1/\sqrt{n}\right)^{n/4} = o(1).$ The result follows. 
	\end{proof}

	\section{Counterexamples}\label{appen}
	In this section we present several graphs used throughout the paper as counter examples. 
	
	\subsection{Concentration}
	We begin two examples showing that the dispersion time doesn't always concentrate. Let $G_1$ be the clique+edge: this is $K_n$ with a extra vertex $v*$ attached by an edge to $v \in K_n$. Let $G_2$ be the clique+hub+edge: a single edge $\{v,v*\}$ attached at $v$ to $h(n)-1$ vertices of the clique $K_{n-2}$.

	\begin{pro}\label{conccounter} 	Let $D^v(G)$ denote either $\tau_{par}^v(G)$ or $\tau_{seq}^v(G)$. Then there exists graphs $G_1$, the clique+edge, and $G_2$, the clique+hub+edge, and $u\in V(G_1), v \in V(G_2)$ such that \[\Pro{D^u(G_1) \leq \mathcal{O}\!\left(\mathbb{E}[D^u(G_1) ]/n \right) } = \Omega(1)\quad \text{and} \quad\Pro{D^v(G_2) \geq \Omega\left(\mathbb{E}[D^v(G_2) ]\cdot n\right)   } = \Omega(1/n).\] 
	\end{pro} 
	\begin{proof}
		Let $G_1$ be the clique+edge. If the parallel or sequential process is started from $v$ then with probability $(1-1/n)^n \approx 1/e$ the vertex $v*$ is not explored in one step and so the process takes $\Omega\left(n^2\right)$ as one of the walks must choose to go back to $v$ and then to visit $v^*$. However with probability $1-(1-1/n)^n \cong 1- 1/e$ one of the $n$ walks hits $v^*$ in the first step and then the process takes $\mathcal{O}(n)$, as is the case with $K_n$. 
		
		Let $G_2$ be the clique+hub+edge. If an IDLA process is started from $v$ then with probability at least $1-\left(1-1/h(n)\right)^n \approx 1-e^{-n/h(n)}$ there is a walker which visits $v*$ in one step. The rest of the graph is essentially a clique and so the process takes $\mathcal{O}(n)$ time. With probability $\left(1-1/h(n)\right)^n\cdot\left(1-1/n \right)^n \sim e^{-n/h(n)-1}$ every walker enters the graph $K_n\backslash N(v)$ and so the process takes an additional $\Theta\left(1/\left(\frac{h(n)}{n}\cdot\frac{1}{n}\cdot\frac{1}{h(n)}\right)\right)$ expected time to cover the graph. Thus $\mathbb{E}[D^v(G_2) ] = \Theta(n)\cdot(1-e^{-n/h(n)}) + \Theta(n^2)\cdot e^{-n/h(n)} $. Choosing $h(n)= n/\log n$ yields $\mathbb{E}[D^v(G_2) ]= \Theta(n)$ and $\Pro{D^v(G_2) \geq \Omega(n^2)   } = \Omega(1/n)$. \end{proof}
	
	\subsection{Least Action Principal}
	
	Continuing our discussion from Section \ref{relwork} we shall show that a least action is violated by a stopping rule on $G_1$, the clique+edge we defined earlier in this section. Let $\xi_x^i=1$ iff the site $x$ is vacant after $i-1$ walkers have settled and $W(X)$ denote the number of walk $X$. The normal ``first vacant site is settled" rule is then $\rho =\inf\left\{t: \xi_{X(t)}^{W(X)}=1 \right\} $.

	\begin{pro}\label{leastact}Define the following stopping rule on $G_1$ \[\tilde{\rho}=\inf\left\{t:\left( t\geq 3 n \log(n)  \textbf{ or } X(t) = v\right) \textbf{ and }\xi_{X(t)}^{W(X)}=1 \right\}\] Then the parallel or sequential process on $G_1$ stopped according to $\tilde{\rho}$ disperses in $\BO{n\log n }$ time. Whereas with the standard stopping rule $\rho$ we have  $t_{seq}(G)= \Omega(n^2)$.  
	\end{pro}	 
	\begin{proof}
		The number of visits to the vertex $v$ at the base of the extra edge $\{v,v^*\}$ is greater than $n\log n $ with probability at least $1- e^{-n}$ by Chernoff bounds. The probability that none of these walks hit $V^*$ is then $(1-1/n)^{n\log(n)}= 1/n$. So $v^*$ is covered by time $3n\log n$ w.p. $1-2/n$ conditional on this the remaining walks settle in time $\BO{n}$. If $v^*$ fails to be covered by time $3n\log(n)$ then the process takes $\BO{n^2}$ by Proposition \ref{conccounter}, the result follows. 
		
		For the standard stopping rule an application of Theorem \ref{partialparupper} shows that we the number of walks is reduced to a sub-polynomial size $k$ in sub-linear time with constant probability. The probability that one of these $k$ random walk hits $v*$ in two steps from $V\backslash\{v,v^*\}$ is $1/n^2$. The probability any of the last $k-1$ hitting $v^*$ before settling (which takes at most $\mathcal{O}(n)$ time) is $o(1)$. Thus occupying $v^*$ is left to the last walk which takes $\Omega(n^2)$ time with constant probability. 
	\end{proof}	
	\subsection{Bounding Dispersion Time from Below by Hitting Time}
	
	The next Proposition, mentioned in Remark \ref{rem:lowbddcounter}, shows that $t_{hit}$ fails as a lower bound for $t_{seq}$. 
	
	\begin{pro} \label{lowbddcounter}
		Fix $0<\epsilon <1/2$ and let $T$ be the complete binary tree on $n$ vertices with a path of length $n^{1/2-\varepsilon}$ attached to the root of the tree at one endpoint. Then
		\[t_{seq}(T) = \mathcal{O}\!\left(n\cdot \log(n)^2 \right) \qquad \text{ and } \qquad t_{hit}(T)= \Omega(n^{3/2-\epsilon}).\]  
	\end{pro}
	\begin{proof} The proof is in the counter examples section of the appendix, Appendix \ref{appen}.
	\end{proof}
	\begin{proof}
		Consider a complete binary tree with $n$ vertices and attach a path of length $k$ by and endpoint to the root, where $1 \leq k = o(\sqrt{n})$. Note that the maximum hitting time in $T$ is $\Theta\left(n \cdot \max\{ k,\log_2(n)\}\right)$, this follows by the commute time identity \cite[Prop. 10.6]{levin2009markov} since effective resistance in a tree is given by graph distance. Considering now the dispersion time, regardless of the source vertex, the root gets at least $\Omega(n)$ visits from $n$ different random walks before all vertices are settled in a binary tree. Every time a walk visits the root, it reaches the endpoint of the path with probability $1/k$ (and in this case, the time to reach the other endpoint is $\Theta(k^2))$. Hence if we consider the Sequential-IDLA, the path of length $k$ is completely covered before the last walk. The expected time for the last walk to settle is then at most the maximum hitting time in the binary tree which is at most $\mathcal{O}(n \log n)$, and with probability at least $1-n^{-2}$, that time is $\mathcal{O}(n \log^2 n)$. By stochastic domination, the time for the $\ell$-th walk to settle for any $1 \leq \ell \leq n-1$ is smaller than that of the last walk. Hence, with high probability all walks are settled after $\mathcal{O}(n \log^2 n)$ time. 
	\end{proof}

	\section{Conclusions}\label{sec:conclusion}
	\subsection{Summary of Our Results}
	The aim of this project is to better understand IDLA processes on finite graphs. The main tool we developed to gain an insight on the processes is the Cut \& Paste bijection. This bijection allows us to study directly the affect of the different scheduling protocols on the random walk trajectories. We use this bijection to couple the various IDLA variants allowing us to order or equate their dispersion times and show that $t_{seq}$ and $t_{par}$ are equal up to a multiplicative factor of order $\log n $.  
	
	In addition to the qualitative information provided by the bijection we also develop upper and lower bounds in terms of graph quantities such as max degree, number of edges, mixing time and hitting times of vertices or sets by a single random walk. These bounds enable us to establish the correct asymptotic order of the dispersion time for the Parallel and Sequential processes on several natural networks. The bounds also provide some tight general bounds in terms of $n$. From our analysis of fundamental networks we conclude that for most natural graphs the dispersion time is of order $t_{hit}$ or $t_{hit}\cdot \log n$ however we present examples where this is very far from the truth.

	\subsection{Further Directions}

	As pointed out earlier, our results establish the correct asymptotic order of the dispersion time for most natural networks. The only exception is the $2d$-grid, where the dispersion time is shown to be between $\Omega(n \log n)$ and $\mathcal{O}(n \log^2 n)$. The known shape theorems for the {\em infinite} 2d-grid, empirical simulations as well as the result for binary trees all strongly suggest the dispersion to be of order $n \log^2 n$. This provides us with the first open problem . 
	
	\begin{pbl}Determine the dispersion time of the $2d$-grid/torus.  
	\end{pbl}
	
	The second main open problem is whether the sequential and parallel dispersion times are of the same order, we know of no graph where this does not hold however it seems hard to prove. 
	\begin{pbl}
		Is it true that for any graph $G$, $t_{par}(G)= \BO{t_{seq}(G)}$?
	\end{pbl}
	In order to prove this result, it might be useful to derive some general lower bounds on the dispersion time, which are in turn interesting and useful in their own right. In particular

	\begin{con}Let $G$ be a connected $n$-vertex graph, then $t_{seq}(G)= \BOhm{n}. $ 
	\end{con} 
	
	The following conjecture is motivated by the idea that when you run $\stp$ algorithm the random walk sections cut and pasted do not have to cover the graph. If true this conjecture would resolve the open problem above some classes of graphs. 
	
	\begin{con}\label{PleqSplusC}Let $G$ be a connected $n$-vertex graph and $t_{cov}(G)$ be the cover time. Then
		\[t_{par}(G) \leq t_{seq}(G) + t_{cov}(G). \] 
	\end{con}

	The counter example to concentration (Proposition \ref{conccounter}) motivates the following open problem.
	\begin{pbl}What conditions must a graph satisfy for the dispersion time to concentrate around its expectation?
	\end{pbl}
	In forthcoming work we examine the total number of steps taken by an IDLA dispersion process and its relation to other graph properties. It might be also worth studying a version of the dispersion process where the origin is sampled uniformly at random for each particle.
	
	\section*{Acknowledgements}
	A.S. is supported by the EPSRC Early Career Fellowship EP/N004566/1. N.R.,T.S. and J.S. are supported by T.S.' ERC Starting Grant 679660 (DYNAMIC MARCH).

	\appendix 
	
	\section{Bounds for Expected Hitting Times of Sets}\label{section:BoundsHitSet}
	We must first state a well known result. 
	
	\begin{lem}[equation (12.11) of \cite{levin2009markov}]\label{lemma:peres1211}
		Consider a lazy random walk on a connected graph, then $\tilde{p}_{u,v}^t \leq \pi(v) + \sqrt{\frac{d(v)}{d(u)}}\lambda_2^t$, where $\lambda_2$ is the second eigenvalue of the associated transition matrix.
	\end{lem}

	We can now prove result which controls $t_{hit}(\pi,S)$ by bounding short-term return probabilities.

	\begin{lem}\label{setestimate}
		Let $G$ be any regular-graph and $S$ be any subset of vertices. Then, for any $v$
		\[
		t_{hit}(v,S) \leq \frac{10}{1-e^{-1}} \cdot \frac{n(1+\lceil \log |S| \rceil)}{(1-\lambda_2) |S|}.
		\]
		Furthermore, suppose that there exists a constants $C>0$ and $\epsilon>0$ such that $p_{u,w}^t \leq \pi(w) + Ct^{-(1+\epsilon)}$ for any pair of vertices $u,w$. Then for any $v$
		\[
		t_{hit}(v,S) \leq \frac{5}{(1-e^{-1})} \cdot \frac{(5C+1)n}{|S|^{\varepsilon/(1+\varepsilon)}}.
		\]
		Both results above extend to almost-regular graphs at expense of a multiplicative $\BO{1}$ factor.
	\end{lem}

	\begin{proof}
		We begin by deriving the first bound. Let $(X_t)_{t \geq 0}$ be a random walk starting from vertex $v$ and let $\tau_S$ be the first time $X_t$ hits the set $S$. We divide time into phases $I_i$ of length $5\tau$ where $\tau = t_{\mix}(1/e)$, i.e. $I_i = \{5(i-1)\tau, \ldots, 5i\tau-1\}$. We count the number of phases needed to reach the set $S$. Suppose that in phases $1,\ldots, i-1$ the walk did not hit $S$. During a phase $I_i$, we let the walk move for $4\tau$ times ignoring if it visits or not the set $S$, and we observe if the walk visited $S$ in the last $\tau$ time-steps of phase $I_i$. Then, independent of everything that happens before time-step $5i\tau$, with probability at least $(1-e^{-1})$ we can couple $X_{4\tau+5(i-1)\tau}$ with the stationary distribution (e.g. Lemma A.5 in \cite{kanade2016coalescence}), hence
		\begin{eqnarray*}
			\mathbf{Pr}_{v}[\tau_S \leq 5i\tau|\tau_S \geq 5(i-1)\tau] \geq (1-e^{-1})\mathbf{Pr}_{\pi}[\tau_S\leq \tau].
		\end{eqnarray*}
		We compute the later probability we define the random variable $Z = \sum_{i=0}^{\tau-1}\mathbf{1}_{\{X_t \in S\}}$ which counts the number of visits of visits to $S$. Then $\mathbf{Pr}_{\pi}[\tau_S \leq \tau] = \mathbf{Pr}_{\pi}[ Z \geq 1]$ and we use the trivial fact that $\mathbf{Pr}_{\pi}[ Z \geq 1] = \mathbf{E}_{\pi}[Z]/\mathbf{E}_{\pi}[Z|Z\geq1]$.
		Clearly $\mathbf{E}_{\pi}[Z] = \tau \pi(S) = \tau |S|/n$. Furthermore
		\begin{align}
		\mathbf{E}_{\pi}[{Z \, \mid \, Z \geq 1}] &\leq \max_{u \in S}  \sum_{t=0}^{\tau} \sum_{v \in S} p_{u,v}^t \leq \sum_{t=0}^{\tau} \min \left\{1,\sum_{v \in S} \left( \frac{1}{n} + \lambda_2^t \right)  \right\}\label{eqn:Zgiven2}\intertext{
			The second inequality holds because $p_{u,v}^t \leq \frac{1}{n}+\lambda_2^t$ for any $u,v$ in a regular graph (Lemma~\ref{lemma:peres1211}). By separating the sum from $t=0$ to $\tau$ at $t=\lceil\log_{\lambda_2}(1/S)\rceil$ and applying $-\log \lambda_2 \leq 1-\lambda_2$ we have } 
		\mathbf{E}_{\pi}[{Z \, \mid \, Z \geq 1}]	&\leq \frac{\lceil\log|S|\rceil}{1-\lambda_2} \cdot 1 +\left( \tau \cdot |S| \cdot \frac{1}{n}  + |S| \cdot \sum_{t=\lceil\log_{\lambda_2}(1/S)\rceil}^{\tau} (\lambda_2)^{t}\right). \nonumber \intertext{Finally bounding the sum by a geometric series,}
		\mathbf{E}_{\pi}[{Z \, \mid \, Z \geq 1}]&=  \frac{\lceil\log|S| \rceil}{1-\lambda_2}  +\tau \cdot |S| \cdot \frac{1}{n}  + \frac{1}{1-\lambda_2}\leq 2\left(\frac{1+\lceil\log |S|\rceil}{1-\lambda_2}\right),\nonumber
		\end{align}since $\tau = t_{mix}(e^{-1}) \leq \frac{1+\log n}{1-\lambda_2}$ by \cite[(12.9)]{levin2009markov}, and $|S|\log n \leq n\log |S|$ for all $|S|\geq 2$. Therefore,
		\begin{align*}
		\mathbf{Pr}_{\pi}[ Z \geq 1]&= \frac{ \mathbf{E}_{\pi}[Z] }{\mathbf{E}_{\pi}[Z \, \mid \, Z \geq 1]} \geq  \frac{\tau \cdot |S| (1-\lambda_2)}{2n(1+\lceil\log |S|\rceil)}.
		\end{align*}
		Denote by $q = (\tau \cdot |S| (1-\lambda_2))(2n(1+\lceil\log |S|\rceil))$. We conclude that 
		$$\mathbf{Pr}_{v}[\tau_S \leq 5i\tau|\tau_S \geq 5(i-1)\tau] \geq (1-e^{-1})q.$$
		From the above, in expectation the walk requires at most $1/(1-e^{-1})q$ phases of length $5\tau = 5t_{mix}$ to finish. Proving the first part of the Lemma.

		The second bound follows the same argument, but replacing $1/n+\lambda_2^t$ by $1/n+Ct^{-(1+\varepsilon)}$, until equation~\eqref{eqn:Zgiven2}. From there split the sum from $1$ to $\tau$ at  $\lfloor|S|^{1/(1+\varepsilon)}\rfloor$, obtaining 
		\begin{align*}
		\mathbf{E}_{\pi}[{Z \, \mid \, Z \geq 1}] 
		&\leq |S|^{1/(1+\epsilon)} \cdot 1 + \tau \cdot |S| \cdot \frac{1}{n} + C|S| \cdot \sum_{t=\lfloor|S|^{1/(1+\varepsilon)}\rfloor}^{\tau} t^{-(1+\epsilon)} \intertext{Now since the sum is less than $\int_{|S|^{1/(1+\varepsilon)}}^{\infty}t^{-(1+\varepsilon)}dt \leq |S|^{\varepsilon/(1+\varepsilon)} $ we have }
		\mathbf{E}_{\pi}[{Z \, \mid \, Z \geq 1}] &\leq |S|^{1/(1+\epsilon)} \cdot 1 + \tau \cdot |S| \cdot \frac{1}{n} +C|S|\cdot |S|^{-\epsilon/(1+\epsilon)}\leq  (5C+1)|S|^{1/(1+\varepsilon)},
		\end{align*}as assumption on the $\tilde{p}_{u,v}^t $ implies $\tau \leq  (4Cn)^{1/(1+\varepsilon)}$ and also $1-\epsilon/(1+\epsilon)=1/(1+\epsilon)$. Hence
		\begin{align*}
		\mathbf{Pr}_{\pi}[{Z \geq 1}]  \geq \frac{\tau |S|/n}{(5C+1)|S|^{1/(1+\varepsilon)}} =   |S|^{\epsilon/(1+\epsilon)} \cdot \frac{\tau}{(5C+1)n}.
		\end{align*}
		The rest of the argument uses the same argument used in the first part of this proof. Recall that a graph is almost-regular if  $\Delta\leq C\cdot \delta  $ holds between the min and max degrees $\delta$ and $\Delta$. In this case since each equation is linear in the constant $C$ we only loose a constant factor.   
	\end{proof}

\end{document}